\documentclass[journal]{IEEEtran}
\usepackage{amsmath,amssymb,amsfonts,epsfig,amsthm,etoolbox,epstopdf}
\usepackage[caption=false,font=footnotesize]{subfig}
\usepackage{dsfont,algorithmic} 
\usepackage[mathscr]{euscript}

\usepackage[square, comma, numbers, sort&compress]{natbib}
\usepackage{multirow}
\usepackage[dvipsnames]{xcolor}

\usepackage{tikz,pgfplots}
\usetikzlibrary{positioning,decorations.pathreplacing}
\usetikzlibrary{fit}
\usetikzlibrary{patterns}
\usetikzlibrary{shapes,snakes}
\usetikzlibrary{arrows, calc}
\usetikzlibrary{plotmarks}

\theoremstyle{definition}
\newtheorem{definition}{Definition}
\newtheorem{lemma}{Lemma}
\newtheorem{theorem}{Theorem}

\newtheorem{proposition}{Proposition}
\newtheorem{approximation}{Approximation}

\theoremstyle{proof}

\everydisplay{\small}
 
\title{Power-Availability-Aware Cell Association for Energy-Harvesting Small-Cell Base Stations}

\author{Fanny Parzysz, \IEEEmembership{Member, IEEE,
} Marco Di Renzo, \IEEEmembership{Senior Member, IEEE,} \\ and Christos Verikoukis, \IEEEmembership{Senior Member, IEEE}
{\color{gray}\thanks{Fanny Parzysz was with the Electronics Department, University of Barcelona, Spain (e-mail: fanny.parzysz@ieee.org).
Marco Di Renzo is with the Laboratoire des Signaux et Syst\`emes,
CNRS, Centrale Sup\'elec, Universit\'e Paris-Saclay, France (e-mail:
marco.direnzo@l2s.centralesupelec.fr).
Christos Verikoukis is with the Telecommunications Technological Centre of Catalonia (CTTC), Castelldefels, Spain (e-mail: cveri@cttc.es).
}}}

\begin{document}

\maketitle

\begin{abstract}
Energy harvesting brings a key solution to the increasing energy bill and environmental concerns but, at the same time, potential energy shortage may deteriorate the network availability.
In this paper, we analyze the performance of off-grid small-cell base stations (scBS) with finite battery capacity and design a new power-availability-aware cell association based on periodical broadcast of the scBS battery level.
Each mobile terminal (MT) targets its own set of available scBSs before association, i.e. the set of scBSs that can guarantee service provided (i) the scBS battery level, (ii) the power required to satisfy a received power constraint at each MT, given the scBS-MT distance and the shadowing attenuation, and (iii) the estimated power consumed to serve other MTs potentially associated to the same scBS, which is computed using stochastic geometry tools. 
Next, we develop for it a tractable performance analysis and derive closed-form expressions for the probability of power outage and the coverage probability.
By dynamically adapting to the fluctuations of the base station battery and the MT received power requirement, the proposed cell association allows a more even distribution of the available energy in the network, brings robustness against harvesting impairment and thereby, significantly outperforms conventional strategies.

\end{abstract}

\begin{IEEEkeywords}
Cellular networks, energy harvesting, cell association, power availability, downlink transmission, power consumption, power outage, stochastic geometry, Markov chain analysis.
\end{IEEEkeywords}

\section{Introduction}

The ongoing deployment of low-powered, low-coverage and low-cost small-cell base stations (scBS) has boosted the performance of cellular networks and is shaping future 5G architecture \cite{Liu2015,Buzzi2016}.
To conform to increasing environmental concerns while providing self-sustainability, energy harvesting technologies have been spurred not only for mobile terminals (MT) but also for access points and base stations. 
Conventional cell association policy is generally designed for on-grid networks, with unlimited power supply, and is not suitable for base stations powered by unpredictable sources. As pointed out in \cite{Liu2015Elsevier,Mao2015}, they do not make sparing use of the energy resource and may rapidly discharge the battery, leading to network availability degradation and reduction in the number of served users.
Thus, it is necessary to design novel strategies tailored to the randomness of the harvesting process and to the battery-limited constraint at base stations.

\textit{Contributions:} We design a new cell association policy dedicated to downlink small-cell networks with off-grid energy-harvesting base stations where: (i) each MT targets its own set of scBSs eligible for association using a \textit{power availability criteria} and (ii) effectively associates to one of them. Then, (iii) each scBS selects for service the associated MTs, in accordance with its actual battery level, and (iv) simultaneously transmits data to all selected MTs.

The proposed power availability criteria aims at adapting each scBS-MT association pair to the scBS battery level and the transmit power required to satisfy a received power constraint at each MT, which is considered as a QoS requirement. In addition to both path-loss and shadowing, it accounts for the power required by \textit{other} MTs potentially associated to the same scBS.
Such criteria remains simple and practically feasible since it is computed independently at each MT and only requires periodical broadcast of the scBS battery level.

To analyze the performance of the proposed cell association, we develop a comprehensive framework using a stochastic geometry approach, together with a Markov chain modeling for the battery. We compute closed-form expressions for both the probability of power outage and the coverage probability.
In addition, a general result is derived to calculate the probability that a base station consumes \textit{in total} exactly $m$ power units. It allows to evaluate the power consumed to transmit data to several users simultaneously , but also to discard the assumption of full-power transmission for the computation of the received interference.
Finally, simulations show that the proposed user association strategy achieves significant outage reduction, is less impacted by a decrease in the energy harvesting rate than conventional policies and is particularly robust against bursty energy arrivals.

\textit{Paper organization:}
The literature on the topic is reviewed in Section \ref{sec:related_work}. The network model and assumptions are described in Section \ref{sec:system_model}. In Section \ref{sec:def_cell_asso}, we present the proposed cell association, formulate the problem raised by the computation of its performance and describe the methodology considered to solve it. The performance of the proposed cell association is analyzed in Sections \ref{sec:performance_analysis_user}, \ref{sec:performance_analysis_BS} and \ref{sec:performance_analysis_final} and simulated in Section \ref{sec:simulations}. Section \ref{sec:conclusion} concludes this paper.

\section{Related work}
\label{sec:related_work}

\subsection{Cell association in energy-harvesting networks}

The optimization of cell association has mainly targeted energy-harvesting networks with centralized control, where global channel state information (CSI) and battery state information (BSI) are available at a central, possibly virtual, node.
A rich literature has proposed algorithms to maximize, for example, the user signal-to-noise ratio, the rate or the resource utilization accounting for the BS battery level and the energy harvesting process \cite{Rubio2014, Liu2014optim, Han2016}. Whereas the proposed solutions provide relevant performance benchmarks, such optimization problems are generally NP-hard, (sub)optimal solutions are found at high computational cost and closed-form performance results are hardly derived for large-scale networks.

Distributed association schemes have been proposed as a response to this issue. Each MT associates to a BS independently of other MTs and based on local CSI/BSI only. Such information shortage, coupled with the uncertainty in the energy arrivals at BSs and in the other MT association decision, renders distributed association more challenging as highlighted in \cite{Maghsudi2015}.
Letting the possibility for MTs to attach to a BS that cannot guarantee service, due to a battery level too low given the traffic demand, leads to severe power outage or service denial. This issue is further discussed in the following and is addressed throughout this paper.

\vspace{-6pt} 
\subsection{Network power availability}

A device equipped with energy harvesting capability may extract energy from a variety of natural or man-made sources, most commonly solar radiations \cite{Mao2015, Ku2015, Lu2015} and more recently, radio-frequency signals \cite{Lu2015, Wang2015,Ghazanfari2015, Maghsudi2016, Krikidis2015, Liu2016, Gu2016}.
Whatever the considered harvesting technology and proposed distribution to model the randomness of energy arrivals, a fundamental aspect of energy harvesting networks with battery-constrained BSs lies in their power availability.  
To guaranty service, one solution lies in considering hybrid base stations, provided with both harvesting facilities and access to the power grid, and analyze the performance of the network by including the power price from electrical grid, as done in \cite{Niyato2012}.
Different from such an approach, we focus on base stations solely powered from energy harvesting. In this case, performance can be significantly improved by controlling the set of BSs eligible for cell association, such that MTs are not necessarily associated with their nearest BS and may be offloaded from BSs with low energy resource to BSs with sufficient battery.

Numerous works have focused on the possibility to activate transmitters or put them in sleep mode depending on their battery level and on energy consumption. Such approach has been considered notably for relay selection \cite{Krikidis2015, Liu2016, Gu2016}, where relays participate to data transmission only if they have sufficient energy stored in their battery. However, these works consider a small-scale network, with a single source-destination pair. Larger-scale networks have be considered in \cite{Dhillon2014_Fundamentals}. In this, MTs associate with the activated base stations providing the highest received signal strength. Such process is opportunistic and varies according to the instantaneous amount of harvested energy. Yet, it relies on a fixed battery threshold and the resulting on / off decision applies equally to all MTs, independently of their power requirements and channel quality. Thus, a BS can be switched off even if it could serve MTs with low power requirements.

With another perspective, biased cell association is analyzed in \cite{Sakr2014, Yu2015}, i.e. the relative cell coverage is extended or shrunk depending on the BS-MT distance. While such solution accounts for the MT power requirement, the proposed bias is the same for all BSs and does not consider the effective amount of available energy, which fails in capturing the impact of the battery fluctuations.
The coupling existing between the user requirement and the battery level regarding cell association is still an open research field in energy harvesting networks, as pointed out in \cite{Mao2015}, and is investigated in this paper.

\vspace{-6pt}
\subsection{Acquisition of the battery level information}

It is widely assumed that data transmission and energy harvesting can occur simultaneously and that only one MT is served at a time. Such model has been considered for uplink transmissions with energy-harvesting MTs in \cite{Sakr2014, Sakr2015} and for downlink transmissions in \cite{Dhillon2014_Fundamentals, Yu2015}.
However, such assumption remains questionable for small-cell BSs simultaneously serving some tens of users and is not applicable once BSI is required at users.
In our framework, targeting the set of base stations eligible for cell association is conditioned by the battery level. This leads to the question of the BSI acquisition.

Several research works have proposed optimal transmission strategies assuming causal and non-causal knowledge of the battery level \cite{Ulukus2015, Ozel2011, Ho2012}.
Yet, in practice, obtaining non-causal information is not feasible and perfect real-time knowledge of the battery level necessitates excessive signaling overhead, low latency information exchange and tight synchronization.
Given the high variability of energy arrivals, traffic demand and user mobility,
continuous broadcast from all BSs of their instantaneous battery level (which varies at each power unit arrival and consumption) is extremely costly.

This raises the double issue of the accuracy and the frequency of battery information exchange. While robust designs have been proposed to account for imperfect battery knowledge, notably for resource allocation schemes \cite{Ku2015, Michelusi2012},
a fully-distributed cell association scheme has been proposed in \cite{Maghsudi2016} where MTs successively associate to base stations that are expected to guarantee a minimum rate for every data packet.
In \cite{Maghsudi2016}, users do not have any information about the effective channel quality, the harvested energy nor the traffic intensity, but only know about their distribution, for which a generic model is proposed.
 
In this paper, we address the issue of BSI acquisition with a novel perspective by considering periodical battery information broadcast. Contrary to real-time information exchange, more than one MT can associate to the same BS based on the same battery information, leading to unique challenge, notably to design an efficient power availability criteria.

\section{System Model and Assumptions}
\label{sec:system_model}

We describe in this section the network model, the assumptions for energy harvesting and the battery model. We invite the reader to refer to Table \ref{table:notation} which summarizes the notation used in this paper.
In particular, the index $k$ will always refer to a small-cell BS, $j$ to a MT and $l$ to a battery state. %

\begin{table*}

\caption{Considered notations}
\label{table:notation}

\centering
\begin{tabular}{|c|c|c|l|}
\hline
\multirow{2}{*}{Indexing} & & $k / j $ & Considered scBS / MT\\
& & $l / l^{\star} $ & Power buffer state / Lowest buffer state providing availability
\\
\hline 
\multirow{8}{*}{PPP / intensity} & \multirow{3}{*}{scBSs} 	& $\Phi_{B}$ / $\lambda_{B}$ 		& scBS location in the network \\[3pt]
 & & $\Phi_{B}^{\star}$ / $\Lambda_{B}$ 		& Set of the $p_{k0}$'s $\forall k$, i.e. Required transmit power from any scBS to MT$^{(0)}$ \\ [3pt]
 & & $\mathcal{A}_j$ / $\Lambda_{B}^{(A)}$ 		& Subset of the $p_{k0}$'s for which scBS$_k$ is available for MT$^{(0)}$ \\[3pt]
\cline{2-4} 

 & \multirow{5}{*}{MTs} & $\Phi_{MT}$ / $\lambda_{MT}$ 		& MT location in the network \\[3pt]
 & & $\Phi_{MT}^{\star}$ / $\Lambda_{MT}$ 		& Set of the $p_{0j}$'s $\forall j$, i.e. Required transmit power from scBS$^{(0)}$ to any MT \\[3pt]
 & & $\mathcal{S}_k$ / $\Lambda_{MT}^{(S)}$ 		& Subset of the $p_{0j}$'s for which scBS$^{(0)}$ serves MT$_j$ \\[3pt]
\hline

\multirow{9}{*}{Power} 	&  \multirow{8}{*}{at scBS$_k$}	& $p_{kj}$		& Power required to send data from scBS$_k$ to MT$_j$\\[3pt]
		&				& $\mathsf{P}_k^{(T)} = \underset{j \in \mathcal{S}_k}{\sum} p_{kj}$			& Total transmit power of scBS$_k$\\[3pt]
		& 				& $ \widetilde{\mathcal{P}}_{k\setminus j}^{(T)} $ 		& Estimated total transmit power from scBS$_k$ to all MTs requiring less $p_{kj}$  \\[3pt]
		
		&				& $\mathsf{P}_k^{(A)}$			& Available power at scBS$_k$\\[3pt]
		&				& $p_{l}^{(\text{cov})}$			& Maximum possible required transmit power when $\mathsf{P}_k^{(A)}=l$ \\[3pt]
		& 				& $\mathsf{P}_k^{(H)}$			& Power harvested by scBS$_k$\\[3pt]

		&				& $\mathsf{P}_{\max}$ / $\mathsf{L}$				& Power buffer capacity in watts / in power units ($\forall k$) \\[3pt]
\cline{2-4}

		&	at MT$_j$	& $\mathsf{P_{Rx}}$					& Received power constraint $\forall j$ \\[3pt]
\hline
\end{tabular}
\end{table*}

\vspace{-6pt}
\subsection{A PPP-based network}

To analyze the performance of the proposed strategy for cell association, we consider a downlink cellular network consisting of one tier of small-cell base stations (scBS) . 
As most of small-cell infrastructures are opportunistically deployed, resulting in irregularly-shaped networks, modeling the node position as random variables allows to analyze the network performance using tools of stochastic geometry and alleviates the need for extensive simulations. In this work, BSs are thus distributed according to an independent homogeneous Poisson point process (PPP), with density $\lambda_{B}$ and denoted as $\Phi_B$. Figure \ref{fig:illustration_association_voronoi} illustrates a potential realization for the BS location, where the cell edges are based on the Voronoi tessellation resulting from associating with the closest BS.

As a widely considered assumption \cite{Dhillon2014_Fundamentals, Sakr2014, Song2014,Yu2015}, mobile terminals are as well distributed according to a homogeneous PPP, independent of $\Phi_B$. It is denoted as $\Phi_{MT}$ and has density $\lambda_{MT}$. This implies that, at each time slot, a random number of MTs is taken from a Poisson distribution and their locations are uniformly distributed over the simulated network area.
Given the extreme densification of next-generation cellular networks, we assume that each MT can connect to more than one base stations, as done for example in \cite{Andrews2014, Bhushan2014}.

\subsection{Channel model for data transmission}
\label{sec:channel_model} 

Downlink transmissions are assumed throughout the analysis presented in this paper.
The channel model accounts for path-loss, shadowing and fast fading. All links are assumed mutually independent and identically distributed. %
First, let $l(r_{k,j}) = \kappa r_{k,j}^{\alpha}$ be the path-loss from scBS$_k$ to MT$_j$, where $r_{k,j}$ is the distance between them, $\kappa$ is the free-space path-loss at a distance of 1m and $\alpha$ is the path-loss exponent.

Second, the shadowing attenuation $\chi_{k,j}$ from scBS$_k$ to MT$_j$  follows a log-normal distribution with p.d.f
\begin{align}
f_{\chi_{k,j}}(w) = \frac{\zeta}{w\sigma \sqrt{2 \pi}} \exp \left( - \frac{\left(10 \log_{10} (w) - \mu\right)^2}{2 \sigma^2}\right)
\end{align}
where $\zeta = 10/ \ln(10)$ and where $\mu$ and $\sigma$ are in dB. %

Finally, the channel from scBS$_k$ to MT$_j$ is subjected to a random complex channel gain, referred as $h_{k,j}$ and whose power gain follows an exponential distribution of parameter $\nu$. The p.d.f. of $\vert h_{k,j} \vert^2$ is expressed by:
\begin{align}
f_{h_{k,j}} (w) = \nu \exp(-w \nu) 
\end{align}

As in \cite{Dhillon2014_Fundamentals, Maghsudi2016}, served MTs are separated in time, frequency or both (OFDMA), such that there is no intra-cell interference. %
The number of sub-channels available at each BS is assumed large enough to accommodate any MT requiring service, as long as sufficient power is available. Such baseline assumption provides a tractable benchmark for performance analysis, where the system is power-limited, rather than load-limited.

\subsection{Assumptions on the power requirement}
\label{sec:on_the_transmit_power}

Cell association policies based on fixed transmit power for any transmission may rapidly discharge the battery and reduce the number of served users \cite{Mao2015}. One way to minimize the power consumption is to make the transmit power as low as possible while satisfying user quality constraints. In this work, we propose an energy-minimization approach for power allocation and consider that scBSs adjust their transmit power to ensure that the average power received at any served user is equal to a given received power constraint $\mathsf{P_{Rx}}$, seen as a QoS requirement. Thereby, a scBS with low available power can nonetheless serve nearby users, with low power requirement. 
Let's define $p_{kj}$ as the transmit power of scBS$_k$ to sent data to MT$_j$, while satisfying $\mathsf{P_{Rx}}$. Thus,
\begin{align}
p_{kj} = \mathsf{P_{Rx}} \; \frac{l(r_{k,j})}{\chi_{k,j}}.
\label{eq:p_kj}
\end{align}
Such value of the BS transmit power is also estimated by each user to target available BSs and associate with one of them. Since fast fading can be hardly tracked for cell association, it is not included in the computation of $p_{kj}$.
Note that $p_{kj}$ can also be interpreted as the transmit power consumed by scBS$_k$ to send data to MT$_j$ over a sufficient period of time such that fast fading is averaged out.

\textit{Remark: } In addition to the transmit power, the overall BS power consumption generally includes the power dissipated in circuitry for data transmission, signal processing, network maintenance and site cooling. It is expressed by:

\begin{align}
P_k^\text{(Total)} = \sum_{j} \left(\eta p_{kj} + P_\text{(dsp)}\right) + P_\text{(idle)}
\label{eq:general_pkj}
\end{align}

\noindent where $\eta$ refers to the RF amplifier loss, $P_\text{(dsp)}$ to the per-user fixed power offset consumed for signal processing and
$P_\text{(idle)}$ to per-BS idle power consumption.
Even if the expression for $p_{kj}$ in Eq. \eqref{eq:p_kj} only accounts for the BS transmit power, the proposed framework can be generalized to more realistic consumption model. Indeed, for both the computation of the availability checking and performance analysis, the factor $\eta$ can be directly included in the received power constraint $\mathsf{P_{Rx}}$ and $P_\text{(idle)}$ can be deduced, without loss of generality, from the battery level which is broadcast  at each time slot for cell association. 
Finally, including $P_\text{(dsp)}$ in the proposed framework only requires some slight modification of the power availability criteria, as described in Subsection \ref{sec:Available_definition}. Yet, its actual value is usually low compared to $\eta p_{kj}$ and $P_\text{(idle)}$, and can be reasonably neglected for performance analysis.

\subsection{A discrete model for the power buffer}  
\label{sec:model_buffer}

We assume as in \cite{Dhillon2014_Fundamentals, Maghsudi2016} that all scBSs are harvesting energy from the environment, e.g. using wind harvester or solar panels. Each is equipped with a battery, or power buffer, with maximal capacity $\mathsf{P}_{\max}$. 
MTs are assumed powered by conventional batteries, charged by users themselves, and their energy limitation is not considered in this paper.

To successfully target the set of available BSs, MTs require to know how much power is stored in their battery. Such information can be easily broadcast by BSs, as part of signaling in control channels. Yet, to limit the resulting extra overhead, we assume that such broadcast is only performed periodically, as depicted in Figure \ref{fig:time_slot}. We refer to the time duration between two battery broadcasts as a time slot. Let $t$ denote for the beginning of a time slot, and $\tau$ for its duration. 
In addition, $\mathsf{P}_k^{(A)}(t)$ refers to the amount of power that is available in the buffer of scBS$_k$ at instant $t$ and broadcast towards MTs for cell association.
 $\mathsf{P}_k^{(H)}(t)$ denotes for the amount of harvested power and $\mathsf{P}_k^{(T)}(t)$ for the power required by associated users (equally the total power consumed by the considered scBS) during this time slot. Then, the following conditions hold:
\begin{align}
\forall t,k \quad \left\lbrace \begin{array}{l}
\mathsf{P}_k^{(A)}(t) = \mathsf{P}_k^{(A)}(t-\tau) - \mathsf{P}_k^{(T)}(t) +  \mathsf{P}_k^{(H)}(t) \\ 
\mathsf{P}_k^{(A)}(t) \leq \mathsf{P}_{\max} \quad \text{and} \quad 
\mathsf{P}_k^{(T)}(t) \leq \mathsf{P}_k^{(A)}(t-\tau)  
\end{array}
\right.
\label{eq:power_conditions}
\end{align}
For sake of readability, the time slot index $(t)$ is dropped in next sections if there is no ambiguity.

\subsubsection{Markov chain model}

We aim to analyze the variation of the battery level $\mathsf{P}_k^{(A)}$ between $t$ and $t + \tau$, $\forall t$. To do so, we propose the following model.
Both the harvested and consumed powers are continuous random variables and so is the amount of power that is stored in the battery. Since such assumption is not tractable for analysis, we consider an approach similar to \cite{Sakr2015, Seunghyun2013} by discretizing the states of the battery into a finite number of levels, as depicted in Figure \ref{fig:battery_model}. As a consequence, the harvested and consumed powers are discretized as well. Note that the stochastic processes of both of them only depend on the buffer state at previous time slot, such that the battery forms a finite-state Markov chain.
We respectively define $\mathsf{L}$ and $\varepsilon$ as the highest battery level and the step size (or power unit), i.e. $\mathsf{P}_{\max} = \varepsilon \mathsf{L}$.
The transition probability matrix of the battery Markov chain is denoted as $\mathbf{P} = \left[ \mathbb{P}_{\overrightarrow{lq}}\right]$, where $\mathbb{P}_{\overrightarrow{lq}}$ stands for the probability to go from state $l$ to state $q$  from one time slot to the next one. %
The probability that the battery has $l$ power units is denoted as $v_l$ and the probability vector of the battery states as $\mathbf{v} = \left[ v_0 v_1 \ldots v_{\mathsf{L}} \right]$.

As stated in Eq. \eqref{eq:power_conditions}, the random variations of $\mathsf{P}_k^{(A)}$ are fully characterized by analyzing the probability to harvest m power units ($\mathsf{P}_k^{(H)}= m$) and to consume n power units ($\mathsf{P}_k^{(T)}= n$), as done in the following two subsections.

\begin{figure*}
	\centering \resizebox{0.9\textwidth}{!}{%
\begin{tikzpicture}

\node [circle,minimum width=0.25cm, draw = black] (State_0) at (0,0) {\scriptsize $l=0$};
\node [circle,minimum width=0.25cm, draw = black, right= 2cm of State_0 ] (State_1) {\scriptsize $l=1$};
\node [circle,minimum width=0.25cm, draw = black, right= 2cm of State_1 ] (State_2) {\scriptsize $l=2$};
\node [right= 3cm of State_2 ] (State_points) {\scriptsize $\cdots$};
\node [circle,minimum width=0.25cm, draw = black, right= 3cm of State_points ] (State_L) {\scriptsize $l=\mathsf{L}$};

\draw[->] (State_0) edge  [loop left=25] node [left] {\scriptsize $\mathbb{P}_{0,0}$} (State_0);
\draw[->] (State_0) edge  [bend left=25] node [below]{\scriptsize $\mathbb{P}_{0,1}$} (State_1);
\draw[->] (State_0) edge  [bend left=25] node [above right]{\scriptsize $\mathbb{P}_{0,2}$} (State_2);
\draw[->] (State_0) edge  [bend left=25] node [above] {\scriptsize $\mathbb{P}_{0,\mathsf{L}}$} (State_L);

\draw[->] (State_1) edge  [loop right=25] node [right] {\scriptsize $\mathbb{P}_{1,1}$} (State_1);
\draw[->] (State_1) edge  [bend left=25] node [below right]{\scriptsize $\mathbb{P}_{1,2}$} (State_2);
\draw[->] (State_1) edge  [bend left=25] node [below left]{\scriptsize $\mathbb{P}_{1,\mathsf{L}}$} (State_L);
\draw[->, dashed] (State_1) edge  [bend left=25] node [above] {\scriptsize $\mathbb{P}_{1,0}$} (State_0);

\draw[->] (State_2) edge  [loop right=25] node [right]{\scriptsize $\mathbb{P}_{2,2}$} (State_2);
\draw[->] (State_2) edge  [bend left=25] node [below]{\scriptsize $\mathbb{P}_{2,\mathsf{L}}$} (State_L);
\draw[->, dashed] (State_2) edge  [bend left=25] node [above right] {\scriptsize $\mathbb{P}_{2,1}$} (State_1);
\draw[->, dashed] (State_2) edge  [bend left=25] node [below right] {\scriptsize $\mathbb{P}_{2,0}$} (State_0);

\draw[->] (State_L) edge  [loop right=25] node [right]{\scriptsize $\mathbb{P}_{\mathsf{L},\mathsf{L}}$} (State_L);
\draw[->, dashed] (State_L) edge  [bend left=25] node [above]{\scriptsize $\mathbb{P}_{\mathsf{L},2}$} (State_2);
\draw[->, dashed] (State_L) edge  [bend left=25] node [above] {\scriptsize $\mathbb{P}_{\mathsf{L},1}$} (State_1);
\draw[->, dashed] (State_L) edge  [bend left=25] node [below] {\scriptsize $\mathbb{P}_{\mathsf{L},0}$} (State_0);

\end{tikzpicture}

 	}
	\caption{Modeling the battery states $\mathsf{P}_k^{(A)}$ of scBS$_k$ as a Markov chain}
	\label{fig:battery_model}
\end{figure*}
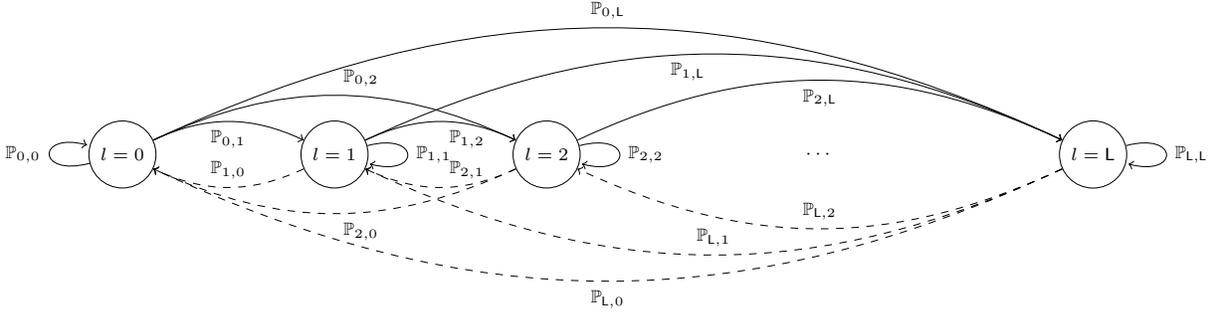

\subsubsection{Model for energy harvesting (battery charge)}

Without loss of generality, a general energy harvesting model is considered, where $N_e$ power units are harvested at a time and where the power arrivals follow a Poisson process of intensity $\lambda_e$. Both parameters $N_e$ and $\lambda_e$ allows to model the burstiness of the harvesting process. The probability $\mathbb{P}_{\text{H}}(m N_e)$ to harvest $m N_e$ units over duration $\tau$ is independent from the current battery state and is given by:
\begin{align}
\mathbb{P}_{\text{H}}(m N_e) = \mathbb{P}\left( \left \lfloor \mathsf{P}_k^{(H)} \right \rfloor = m N_e\right)
= \frac{\left(\lambda_e \tau \right)^m}{m!} \exp \left(- \lambda_e \tau \right).
\label{eq:prob_H}
\end{align}
	As shown in \cite{Miozzo2013, Ku2015}, solar energy arrivals are efficiently modeled as stochastic Markov processes, and more specifically Poisson process \cite{Wang2015}. In this case, the parameters $N_e$ and $\lambda_e$ are determined by the quantity of sunlight, cloud coverage, air  density, temperature, but also the size of the photo-voltaic panel \cite{Wang2015}.
We also highlight that the analysis proposed in this work is valid for other harvesting processes, by considering other function  $\mathbb{P}_{\text{H}}$.

\subsubsection{Model for energy consumption (battery discharge)}
\label{sec:frame_model}

To model the battery discharge, a birth-death process has been proposed in \cite{Dhillon2014_Fundamentals, Krikidis2015, Gu2016, Liu2016}, where users consume exactly one power unit and where the battery level is updated at each new data transmission. Yet, such a model cannot be applied in our framework. First, each user has a different power requirement, as stated in Eq. \eqref{eq:p_kj}, implying that it is not sufficient to consider the number of users currently being served to model the battery discharge. 
Moreover, in our case, updated battery levels are \textit{not} broadcast after each new user association, such that the battery decrease potentially accounts for transmission towards \textit{more than one} user.

To compute the probability to consume $n$ power units  ($\mathsf{P}_k^{(T)}= n$) and characterize the transition probability matrix  $\mathbf{P}$, the power consumed by \emph{all} users associated during this time slot should be analyzed\footnote{ For analysis, we assume that data transmission effectively occurs in the next time slot, as shown in Figure \ref{fig:time_slot}.}.  
The amount of elementary power $p_{kj}$ consumed to send data from a given scBS$_k$ to MT$_j$ is rounded up to the nearest battery unit and the probability $\mathbb{P}_{\text{T}}(n \; \vert \; l)$ to consume a \emph{total} of $n$ power units given that scBS${_k}$ has $l$ power units in its battery is equal to $ \mathbb{P}\left( \mathsf{P}_k^{(T)}(t) \simeq \sum_j \left \lceil p_{kj} \right \rceil = n \; \vert \; l \right)$. Its computation is one of the main issues solved in this paper. %

\begin{figure}
	\centering \includegraphics[width = 0.9\columnwidth]{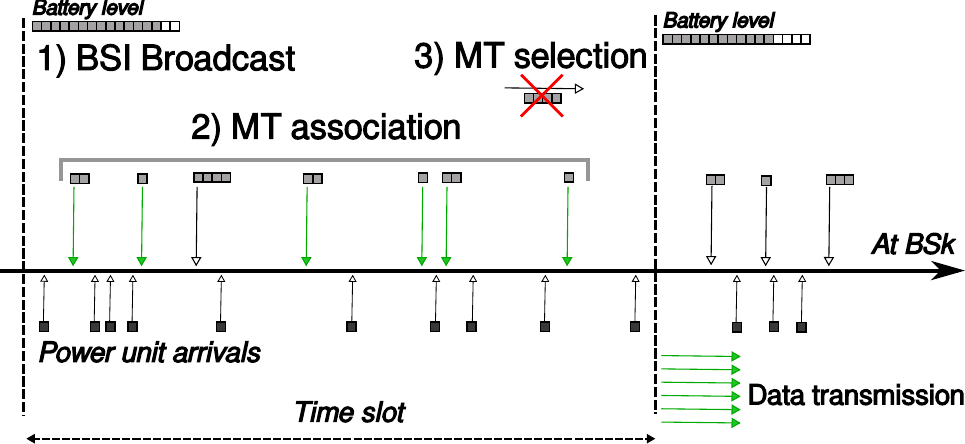}
	\caption{Time slot model}
	\label{fig:time_slot}
\end{figure}

\textit{Remark 1:} The arrival order within a time slot is not considered and MTs are selected in ascending order of their power requirement.

\textit{Remark 2:} Increasing the battery broadcast period, i.e. the time slot duration $\tau$, suggests a higher $\lambda_{MT}$, since more MTs can request cell association during the same slot. This is discussed in Section \ref{sec:simulations}.

\begin{figure*}
	\centering
	\subfloat[$\mathsf{P_{Rx}} =$-60dBm, $N_e = 1$, $\lambda_e = 5\% \; \mathsf{L} / N_e$]{ \includegraphics[width=0.32\textwidth]{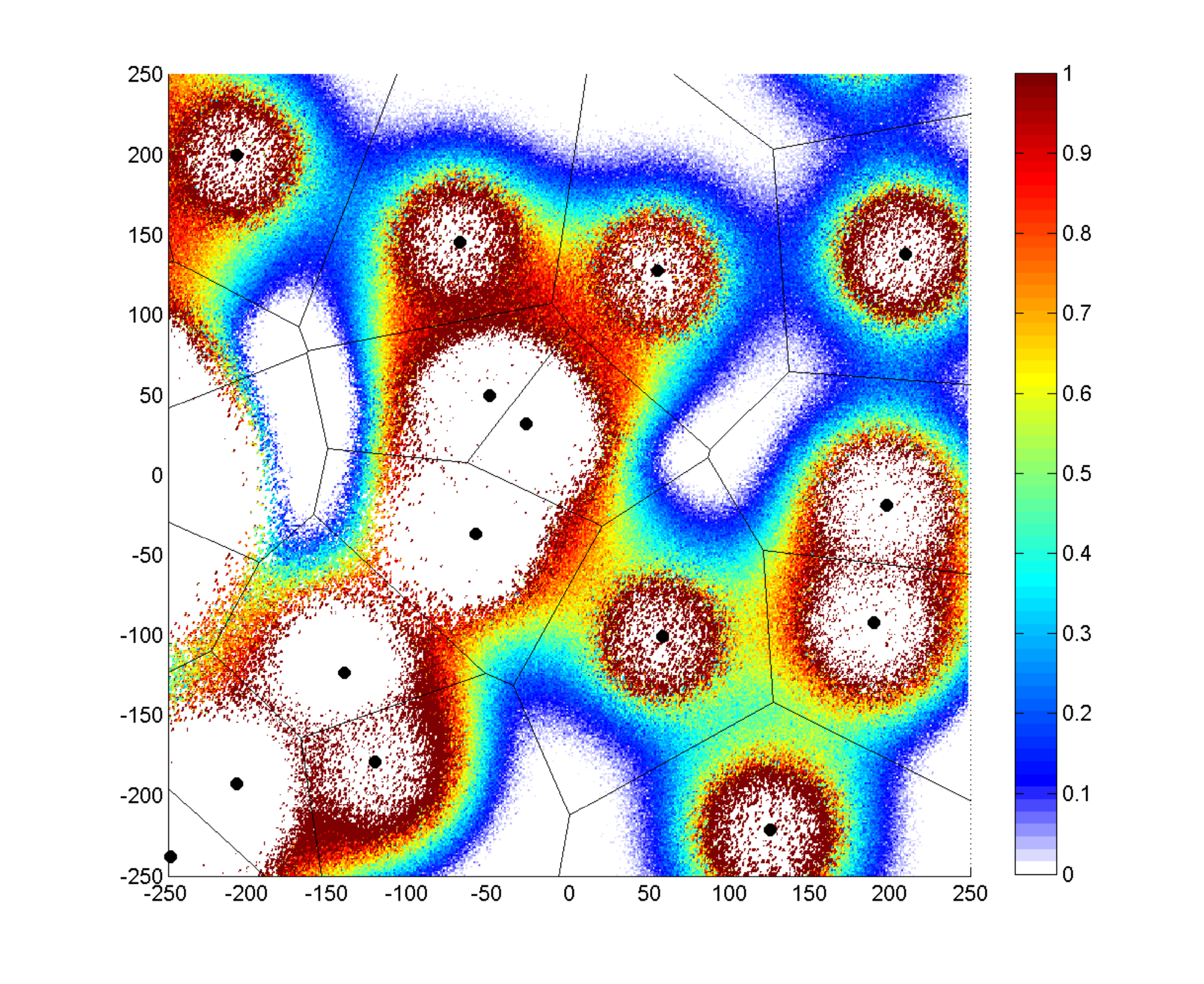}}
	\subfloat[$\mathsf{P_{Rx}} =$-65dBm, $N_e = 1$, $\lambda_e = 10\% \; \mathsf{L} / N_e$]{ \includegraphics[width=0.32\textwidth]{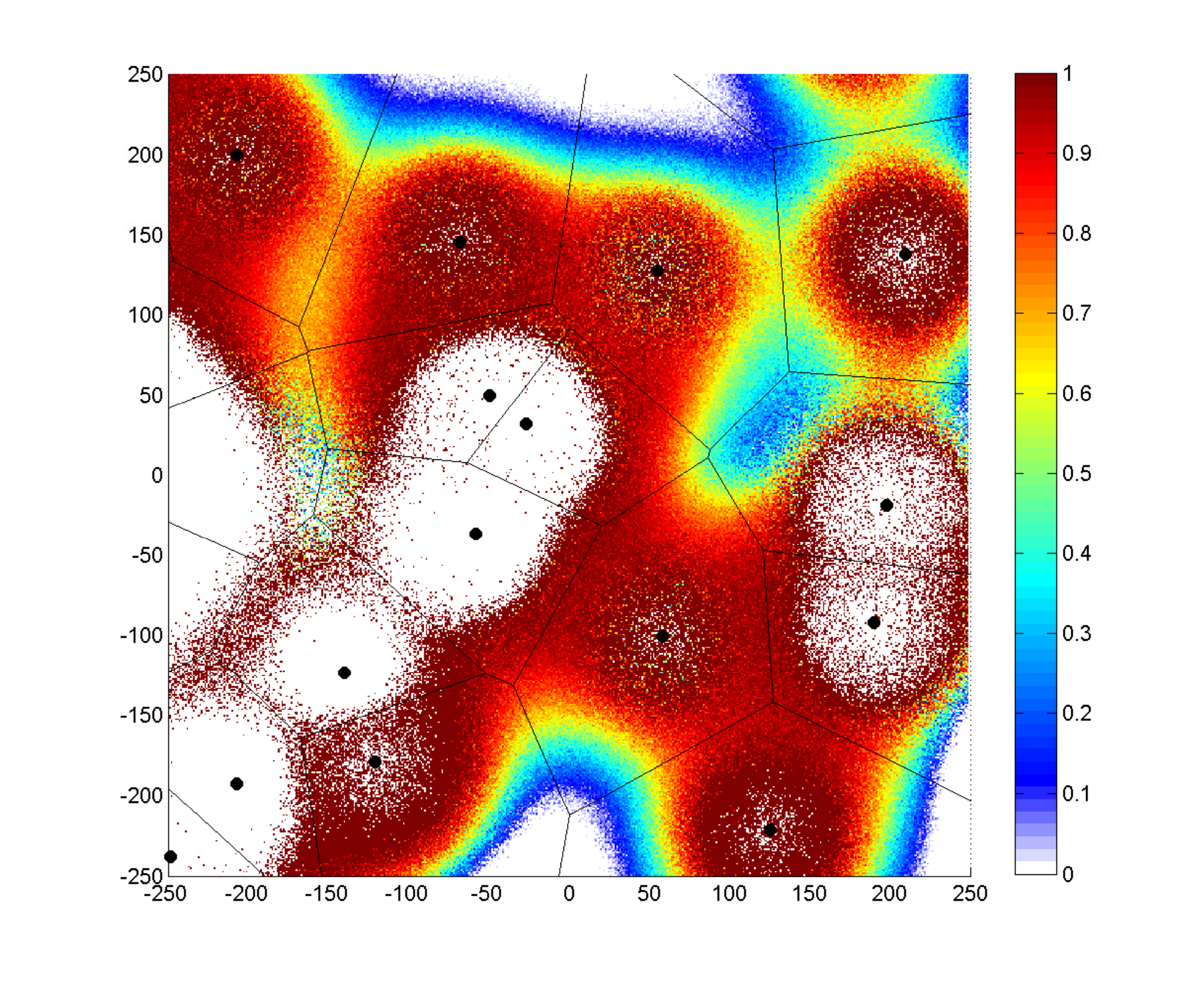}}
	\subfloat[$\mathsf{P_{Rx}} =$-65dBm, $N_e = 80$, $\lambda_e = 10\% \; \mathsf{L} / N_e$]{ \includegraphics[width=0.32\textwidth]{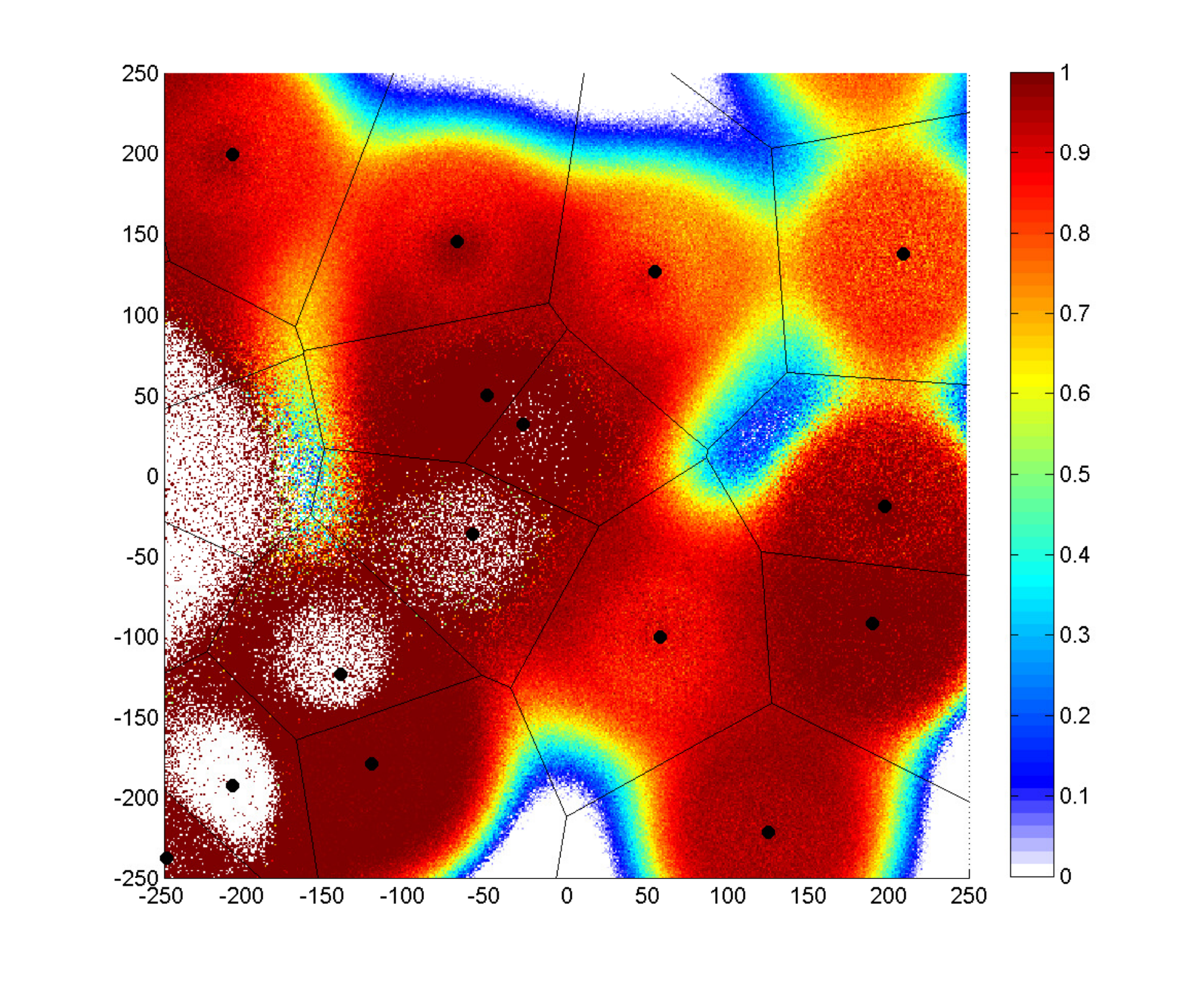}}
	\caption{Gain in the probability of outage
	}
	\label{fig:illustration_association_voronoi}
\end{figure*}

\section{Proposed Cell Association and Methodology}
\label{sec:def_cell_asso}

This section describes the proposed energy-aware cell association with BS availability checking and presents the methodology considered for performance analysis.

\subsection{Overview of the proposed cell association}
\label{sec:overview}

In a network deprived of access to the power grid, a BS with low battery may not be able to guaranty service.
By knowing the BS battery state and by comparing such value with the energy required to meet the received power constraint $\mathsf{P_{Rx}}$, users can avoid associating with a nearby BS in energy shortage and prefer another with higher battery level, even if such scBS is farther. 
The proposed cell association accounts for both aspects of the scBS available power and the MT required transmit power, as expressed in Eq. \eqref{eq:p_kj}, and can be interpreted as an energy-oriented offloading technique. We further assume that scBSs are not interacting nor cooperating.

\subsubsection{Procedure} The proposed strategy is depicted in Figure \ref{fig:time_slot} and is performed in three steps, over each time slot.

\textit{- At the beginning of a time slot}, scBSs broadcast their battery state, for example by using the  PDCCH (Physical Downlink Control Channel) of the LTE frame.
	
\textit{- During this time slot}, when a new MT requests service, cell association is performed in two steps: 1) Power availability and 2) Effective cell association. First, this MT checks if the transmit power required to satisfy the received power constraint $\mathsf{P_{Rx}}$ is compatible with the battery level of the neighboring scBSs, as announced by broadcast. By doing so, each MT targets a subset of scBSs which are eligible, or available, for cell association. Second, each MT associates to one of the available small-cell BSs. It chooses the one which consumes the least transmit power while meeting $\mathsf{P_{Rx}}$.

\textit{- At the end of the time slot}, a third step specific to the considered energy-constrained context must be considered before data transmission. Due to power limitation, a scBS may not be able to serve all associated MTs, if too many are associated to it despite power availability checking.
Thus, each scBS proceeds to MT selection. Unselected MTs are denied service and wait for the next time slot, with updated information on the battery state, to target newly available BSs and perform once again association. Nevertheless, the proposed availability criteria results in a very low probability to be denied service once associated, as shown in Section VIII.

\subsubsection{Principal advantages}

Figure \ref{fig:illustration_association_voronoi} illustrates the gain in the probability of power outage that is achieved by the proposed energy-aware cell association (namely, scheme "$\mathcal{A}$") over conventional association without availability checking (namely, scheme "\textit{w/o}$\mathcal{A}$"). For each MT location in the network, we compute the ratio $\left( \mathbb{P}_{\text{out}}^{(\text{w/o}\mathcal{A})} - \mathbb{P}_{\text{out}}^{(\mathcal{A})} \right) / \mathbb{P}_{\text{out}}^{(\text{w/o}\mathcal{A})}$. As observed in Figures (a) and (b), the proposed strategy significantly improves the performance at cell edge, even when the energy harvesting rate is low compared to the cell coverage (Fig. (a)). When the energy arrival is bursty ($N_e = 80$ - Figure (c)), performance gain is generally higher and obtained in a larger part of the network. Indeed, the availability criteria follows the BS battery variations, such that users are dynamically distributed among neighboring scBSs as a function of energy arrivals. This is further discussed in Section \ref{sec:sim_results}.

This reduction in the probability of power outage provides fairness in the user service, as less energy is consumed to serve cell-center users, implying that more energy is available to serve cell-edge users with the same received power constraint.
Contrary to network-optimized association policies, the proposed strategy can be easily implemented and does not require extra signaling overhead (excepting BS battery broadcast), e.g. link quality exchange, nor tight BS synchronization/cooperation, nor heavy computation at a central node, which all consume significant energy as well.

\subsection{Description of the proposed cell-association} 

\subsubsection{Set of available small-cell BSs}
\label{sec:Available_definition}

The power availability criteria characterizes the set of available scBSs and each MT determines its \emph{own} subset of available scBSs as follows:

\begin{definition}
The subset of scBSs declared available by MT$_j$ is denoted as $\mathcal{A}_j$ (with $\mathcal{A}$ for Available) and is given by:
\begin{align*}
& \mathcal{A}_j \triangleq \left \lbrace \text{scBS}_k \in  \Phi_B:  p_{kj} + \widetilde{\mathsf{P}}_{k\setminus j}^{(T)}  \leq \mathsf{P}_k^{(A)}   \right \rbrace 
\end{align*}
where $\widetilde{\mathsf{P}}_{k\setminus j}^{(T)}$ refers to the \emph{estimated} transmit power required by other MTs associated to scBS$_k$ and selected before MT$_j$ given the MT selection rule (Definition \ref{def:selection_rule}). It is defined as the expected sum of the power required by any other MT$_i$ satisfying $p_{ki} < p_{kj}$, i.e.
\begin{align*}
\widetilde{\mathsf{P}}_{k\setminus j}^{(T)} &\triangleq \mathbb{E} \left[ \sum_{i \in \Phi_{k\setminus j}} p_{ki}\right]
\quad \text{with} \quad
\Phi_{k\setminus j} \triangleq \left \lbrace \text{MT}_i \in  \Phi_{MT}:  p_{ki} < p_{kj}  \right \rbrace .
\end{align*}
The estimate $\widetilde{\mathsf{P}}_{k\setminus j}^{(T)}$ is computed later on in Eq. \eqref{eq:estimate_Psetminus0}.
\label{def:available}
\end{definition}

\emph{Insights for Definition \ref{def:available}:}
Several power availability criteria are possible for the scBS pre-selection rule. The simplest one would be to check if the required power $p_{kj}$ does not exceed the battery level, i.e. $p_{kj} \leq \mathsf{P}_k^{(A)}$ and $\widetilde{\mathsf{P}}_{k\setminus j}^{(T)}=0$.
Whereas such criteria is valid if the battery level is known at MTs in real-time, this is not adapted to the considered multi-user transmission context with periodical broadcast. Given the power constrains in Eq. \eqref{eq:power_conditions}, the \emph{total} power of served MTs, i.e. $\mathsf{P}_{k}^{(T)} = p_{kj}+\sum_{i \in \mathcal{S}_k} p_{ki}$, with $\mathcal{S}_k$ the set of MTs served by scBS$_k$ , has high probability to exceed $\mathsf{P}_k^{(A)}$ with such availability criteria, leading to severe power outage. 
To avoid associating to a scBS that cannot guarantee service, a valid power availability criteria should accommodate the transmit power required by the other users served by scBS$_k$, i.e. $\sum_{i \in \mathcal{S}_k} p_{ki}$, in addition to the power $p_{kj}$ required by MT$_j$. However, the total consumption $ \mathsf{P}_{k}^{(T)}$ depends on the cell association of all MTs of the network which cannot be known in advance for computing the power availability criteria (non-causal information). An estimate of $\mathsf{P}_{k}^{(T)}$ is thus required.
If overestimated, the set of available scBSs may be excessively reduced.
If underestimated, a scBS may be associated to more MTs than it can actually serve given its available power, as for the criteria $p_{kj} \leq \mathsf{P}_k^{(A)}$.

As scBSs serve MTs in ascending order of required powers, it is sufficient to account only for the total transmit power of the served MTs that require less power than $p_{kj}$. In the proposed availability criteria, MT$_j$ assumes that any other MT which requires less power than $p_{kj}$ is associated with scBS$_k$, regardless of their effective cell association. The estimate $\widetilde{\mathsf{P}}_{k\setminus j}^{(T)}$ only requires to know the density of MTs $\lambda_{MT}$ and the battery level $\mathsf{P}_k^{(A)}$. 
Its value can be computed using stochastic geometry tools as follows:
\begin{align}
\widetilde{\mathsf{P}}_{k\setminus j}^{(T)}&= \lambda_{MT} \Upsilon \frac{2 / \alpha} {2 / \alpha +1} \left( p_{kj} \right)^{\frac{2}{\alpha}+1}
\label{eq:estimate_Psetminus0}
\\
\text{where} \quad & \left \lbrace \begin{array}{l l}
\Upsilon &= \pi \left(\frac{1}{\mathsf{P_{Rx}} \kappa}\right)^{\frac{2}{\alpha}} \exp \left(  \frac{2/\alpha}{\zeta} \mu + \frac{1}{2} \left(\frac{2/\alpha}{\zeta}\right)^2 \sigma^2 \right) 
\\
\zeta &= 10/ \ln(10)
\end{array} \right. \nonumber
\end{align}
The proof of Eq. \eqref{eq:estimate_Psetminus0} is given in Appendix \ref{appendix:estimate_Psetminus0}.

\textit{Remark: } Assuming a generalized power consumption model as described by Eq. \eqref{eq:general_pkj}, the per-user power consumed for signal processing at the serving BS can be included in the power availability criteria by adding $P_\text{(dsp)} \Lambda_{MT}(p_{kj})$ to $\widetilde{\mathsf{P}}_{k\setminus j}^{(T)}$, where $\Lambda_{MT}(p_{kj})$ is the average number of MTs requiring less than $p_{kj}$. Yet, this does not modify the proposed analysis.

\vspace{5pt}
\subsubsection{Effective cell association}
\label{sec:CellAss}

Once available scBSs have been targeted, each MT effectively associates to one of them.

\begin{definition}
A mobile terminal MT$_j$ associates with the available base station which consumes the least transmit power to satisfy the received power constraint $\mathsf{P_{Rx}}$. It is denoted as scBS$_0$, with:
\begin{align}
\text{scBS}_0 = \underset{\text{scBS}_k \in \mathcal{A}_j} {\arg \min} \left\lbrace \mathsf{P_{Rx}} \; \frac{l(r_{k,j})}{\chi_{k,j}}   \right \rbrace
= \underset{\text{scBS}_k \in \mathcal{A}_j} {\arg \min} \left\lbrace p_{kj}   \right \rbrace
\end{align}
\vspace{-5pt}
\label{def:closest_power}
\end{definition}
This cell association is equivalent to associating with the small base station for which the received 
power is maximal, for a fixed transmit power equal to $\frac{1}{\mathsf{P_{Rx}}}$.

\vspace{5pt}
\subsubsection{MTs selection rule}
\label{sec:MTselection}

To maximize the number of served MTs while satisfying the received power constraint for all of them, we consider the following approach: 
\begin{definition}
Each small-cell BS selects mobile terminals to be served among the associated ones and in ascending order of their power requirement till all associated users are selected or till the overall power consumption reaches the battery level. This subset of MTs is denoted as $\mathcal{S}_k$ for scBS$_k$.
\label{def:selection_rule}
\end{definition}

We will show that the probability to be associated but rejected is negligible for the proposed cell association.

\textit{Remark: } Such MT selection strategy allows hard decision for the power availability criteria, which is not permitted by a random MT selection, where a small-cell base station can be declared available only with a given probability.

\subsection{Problem formulation and methodology for analysis}

To analyze the performance of the proposed energy-aware cell association in terms of power outage and coverage probability, it is necessary to understand how the availability criteria of Definition \ref{def:available}, the effective association of Definition \ref{def:closest_power} and the MTs selection rule of Definition \ref{def:selection_rule} affect the set of served MTs and the battery level and reciprocally, how the battery level determines the set of scBSs eligible for association.

Although Monte-Carlo simulations offer wide possibilities and numerous variable factors for refined performance analysis, they suffer from two main drawbacks. First, they are time-consuming and their accuracy is significantly affected by the number of samples considered for simulations. Second, they confine the obtained conclusions only to the considered parameter settings and new sets of simulation are required for performance generalization, since parameters cannot be modified retroactively. This motivates us to compute closed-form expressions for performance analysis, which is the object of the remainder of the paper.

The analytical approach considered in \cite{Dhillon2014_Fundamentals} to compute the power outage and coverage probabilities is based on the expected cell coverage, average number of served MTs and average transmit power consumption. Yet, it cannot be applied to the proposed power availability checking which accounts for their instantaneous fluctuations. Instead of averaging, the probability mass functions (pmf) of the battery states and of the number of power units consumed at each BS transmission are computed in this paper. To this end, we point out that the Markov chain formed by the battery levels converges to a unique stationary distribution solving:
\begin{align}
\left \lbrace \begin{array}{ll}
\mathbf{v} &= \mathbf{v} \mathbf{P} \\
1 &= \mathbf{v} \mathbf{1}
\end{array}
\right.
 \label{eq:steady_state} 
\end{align}
where $ \mathbf{1}$ is a all-ones vector and $\mathbf{v} = \left[ v_0 v_1 \ldots v_{\mathsf{L}} \right]$ is the steady-state battery probability vector, of dimension $1\times (\mathsf{L}+1)$.

\begin{proof}
	The considered Markov chain is aperiodic and irreducible ($\mathbb{P}_{i,j}>0, \forall i,j$, i.e. all the states communicate), with finite state space. 
	Next, $\forall l$, $\sum_q \mathbb{P}_{\overrightarrow{lq}} = 1$ and, given that the energy harvesting and energy consumption processes are distinct and independent,  $\mathbb{P}_{\overrightarrow{lq}} \neq \mathbb{P}_{\overrightarrow{ql}}$, $\forall l,q$, implying that the transition probability matrix $\mathbb{P}$ is not symmetric. Consequently, it is row stochastic and not doubly stochastic.
\end{proof}

The proposed analysis aims at solving such system.
The main challenge lays in the computation of the transition probability matrix of the battery Markov chain $\mathbf{P} = \left[ \mathbb{P}_{\overrightarrow{lq}}\right]$, of dimension $(\mathsf{L}+1)\times (\mathsf{L}+1)$, which is itself a function of $\mathbf{v}$ and renders the system of equations in \eqref{eq:steady_state} non-linear.
To solve such issue, we consider the following methodology. First, we analyze the conditions for availability and characterize the set of available BSs as a function of the probability vector $\mathbf{v}$ (Subsection \ref{sec:Availability_analysis}). Next, we deduce the sets of MTs associated and served by a given BS (Subsection \ref{sec:Service_analysis}). Based on the densities of available scBSs and served MTs, we compute $\mathbb{P}_{\overrightarrow{lq}}, \; \forall l,q$ as a function of $\mathbf{v}$ and propose an algorithm to solve Eq. \eqref{eq:steady_state} (Subsection \ref{sec:Battery_analysis}). Finally, performance are computed (Subsection \ref{sec:performance_analysis_final}).

\section{Analysis of the proposed cell association: \\ the User's perspective} 
\label{sec:performance_analysis_user}

This section is dedicated to the characterization of the set of available BSs and the probability of cell association.

\subsection{Analysis of the set of available scBSs}
\label{sec:Availability_analysis}

\textit{Generalities: }For analysis, the Slivnyak theorem allows considering a single typical small-cell base station, denoted as scBS$^{(0)}$, without loss of generality. Let's define $\Phi_{MT}^{\star}$ as the point process of the $p_{0j}$'s $\forall j$, for this typical scBS. From the properties of displaced PPPs, scBS$^{(0)}$ sees the $p_{0j}$ as distributed according to a non-homogeneous PPP on $\mathbb{R}^+$ with intensity $\Lambda_{MT}(p) = \lambda_{MT} \Upsilon p^{\frac{2}{\alpha}}$, as proved in Appendix \ref{appendix:estimate_Psetminus0}.
Similarly, a typical mobile terminal, denoted as MT$^{(0)}$, can be considered for analysis. It sees the $p_{k0}$'s, $\forall k$ as distributed according to a non-homogeneous PPP, denoted as $\Phi_{B}^{\star}$, and of density  $\Lambda_{B}(p)= \lambda_{B} \Upsilon p^{\frac{2}{\alpha}}$.

Given the power availability criteria of Definition \ref{def:available} and the estimated power consumption of Eq. \eqref{eq:estimate_Psetminus0}, 
the typical user MT$^{(0)}$ assumes that any other MT which requires less power than $p_{k0}$ is associated with scBS$_k$.
Denoting $\mathsf{P}_k^{(A)}= l$ power units, the availability criteria is equivalently defined as:
\begin{align}
& p_{k0} + \widetilde{\mathsf{P}}_{k\setminus 0}^{(T)}  \leq \mathsf{P}_k^{(A)}
\quad \Leftrightarrow \quad g_{A}(p_{k0}) \leq l
\nonumber \\
& \text{with } \quad g_{A}(p) = p + \lambda_{MT} \Upsilon \frac{2 / \alpha} {2 / \alpha +1}  p^{\frac{2}{\alpha}+1} 
\label{eq:g_A}
\end{align}
Since $g_{A}(p+1) - g_{A}(p)> 1$, increasing the required transmit power by one power unit necessitates a much higher increase of the power available in the power buffer to account for all other MTs requiring only $p$.

We define $p_{l}^{(\text{cov})}$ as the maximum power  that can be required by a MT served by a scBS having $l$ power units in its battery, i.e. the maximum power satisfying the availability criteria:
\begin{align}
p_{l}^{(\text{cov})} = g_{A}^{-1}(l)
\label{eq:p_cov}
\end{align}
$p_{l}^{(\text{cov})}$ can be interpreted in terms of power coverage. Any MT$_j$ for which $p_{kj} \leq p_{l}^{(\text{cov})}$
 declares scBS$_k$ available for data transmission. Reciprocally, the MTs served by scBS$_k$ are solely located in this subset. Note that scBS$_k$ is necessarily unavailable for any MT$_j$ such that $p_{kj} > p_{\mathsf{L}}^{(\text{cov})}$. 

Likewise, the lowest buffer state $l^{\star}$ providing availability of scBS$_k$ for MT$^{(0)}$ is defined as:
\begin{align}
l^{\star}-1 < g_{A}(p_{k0}) \leq l^{\star} \quad \Leftrightarrow \quad
p_{l^{\star}-1}^{(\text{cov})} < p_{k0} \leq p_{l^{\star}}^{(\text{cov})} 
\label{eq:l_star}
\end{align}
with $p_{k0}$ the required power to transmit data from scBS$_k$ to MT$^{(0)}$.
Using such notation, scBS$_k$ is declared available if its battery state $l$ is greater than $l^{\star}$. This occurs with probability:
\begin{align}
\mathbb{P}_{\text{B}}^{(\text{A})} (p_{k0}) = 
\underset{l = l^{\star}}{\overset{\mathsf{L}}{\sum}} v_l
\end{align} 
where $v_l$ is the probability to be in state $l$. Such probability allows to compute the intensity of available BSs.

Given that both scBSs and MTs are independently located, that energy is harvested independently at each base station and that the power availability checking is performed independently at each MT as well, each BS of $\Phi_{B}^{\star}$ is declared available by MT$^{(0)}$ with  probability $\mathbb{P}_{\text{B}}^{(\text{A})} (p_{k0})$.

\begin{lemma}
The subset $\mathcal{A}_0$ of small-cell base stations available for MT$^{(0)}$ forms a PPP on $\left[0,p_{\mathsf{L}}^{(\text{cov})}\right]$, with density
\begin{align}
& d\Lambda_{B}^{(A)} (p) = d\Lambda_{B} (p) \mathbb{P}_{\text{B}}^{(\text{A})} (p)
 = d\Lambda_{B} (p) \underset{l = l^{\star} }{\overset{\mathsf{L}}{\sum}} v_l
\nonumber
\\
& \Rightarrow \quad \Lambda_{B}^{(A)} (p) = \underset{l=0}{\overset{l^{\star}-1}{\sum}} v_l \Lambda_{B} (p_{l}^{(\text{cov})}) +  \underset{l = l^{\star} }{\overset{\mathsf{L}}{\sum}} v_{l} \Lambda_{B} (p)
\label{eq:Lambda_B_A}
\end{align}
with $ l^{\star}$  depending on $p$ as in Eq. \eqref{eq:l_star} and $\Lambda_{B}(p) = \lambda_{B} \Upsilon p^{\frac{2}{\alpha}}$.
\end{lemma}
\begin{proof}
This Lemma directly follows from the thinning property of PPPs and from \cite[Prop.1.3.5]{book_sto_geo}.
\end{proof}

This expression of $\Lambda_{B}^{(A)}$ can be interpreted as follows. Any scBS for which the transmit power required by MT$^{(0)}$ is equal to $p \in \left[0, p_{l^{\star}}^{(\text{cov})} \right]$ is declared available if the power buffer state is greater than $l^{\star}$. This gives the term $ \Lambda_{B} (p) \underset{l = l^{\star} }{\overset{\mathsf{L}}{\sum}} v_{l} $. However, any scBS for which $p \in \left[0, p_{l^{\star}-1}^{(\text{cov})} \right]$ can be declared also available if the power buffer is in state $l^{\star}-1$. This gives the additional term $v_{l^{\star}-1} \Lambda_{B} (p_{l^{\star}-1}^{(\text{cov})})$, etc.

\subsection{Probability to be associated with the typical scBS}
\label{sec:Association_analysis}

By Definition \ref{def:closest_power}, the probability $\mathbb{P}_B^{(\text{Ass})}(p \vert l)$ that a MT requiring a power $p$ is associated with the typical small-cell base station sc-BS$^{(0)}$, given that $l$ power units are available in its battery, is equal to the probability that scBS$^{(0)}$ is the "closest" available base station in terms of transmit power $p$. 
With $\Lambda_{B}^{(A)} (p)$ the density of available scBSs, it follows from the void probability of PPPs that
\begin{align}
\mathbb{P}_B^{(\text{Ass})}(p \vert l) = \left \lbrace
\begin{array}{l l}
\exp \left(- \Lambda_{B}^{(A)} (p) \right) & p \leq p_{l}^{(\text{cov})}
\\
0 & \text{otherwise}
\end{array}
\right.
\end{align}

\section{Analysis of the proposed cell association: \\ the BS's perspective} 
\label{sec:performance_analysis_BS}

This section switches to the perspective of a BS and characterizes the set of served MTs and the battery state.

\subsection{A useful theorem: sum of transmit power }

Before moving onto the analysis of the proposed cell association, we propose a general and useful theorem:
\begin{theorem}
	Assume that the set of the transmit powers $p_{0j}$ from the typical base station scBS$^{(0)}$ to its served MTs form a PPP denoted as $\Phi$ and of density $\Lambda(p)$. In addition, assume that $p_{0j} \leq P \in \mathbb{N}, \forall j $. Then, the probability of the discretized total power consumption at scBS$^{(0)}$ denoted as $\mathbb{P}_{\Sigma} \left(m,\Lambda,P\right)=
	\mathbb{P}\left( \underset{j \in \Phi}{\sum} \lceil p_{kj} \rceil = m \right) $ is found recursively as follows:
	\begin{align*}
	& \left \lbrace \begin{array}{rl}
	\mathbb{P}_{\Sigma} \left(0,\Lambda,P\right) %
	&=\exp \left(-\Lambda (P) \right) 
	\\
	\mathbb{P}_{\Sigma} \left(m,\Lambda,P\right)&= \underset{q=1}{\overset{m^{\star}}{\sum}}  \frac{q}{m} \; C_q \;  \mathbb{P}_{\Sigma}(m-q,\Lambda,P)
	\end{array} \right.
	\end{align*}
	where  $\quad C_{q} = \Lambda (q) - \Lambda (q-1)$ and $m^{\star} = \min(m, P)$.
	\label{prop:useful_lemma}
\end{theorem}
\begin{proof}
	See Appendix \ref{appendix:useful_lemma}.
\end{proof}

In Theorem \ref{prop:useful_lemma}, $C_q$ refers to the average number of MTs requiring exactly $q$ units of energy for data transmission (after rounding up). The probability $\mathbb{P}_{\Sigma} \left(0,\Lambda,P\right)$ that the scBS does not consume energy at all is equal to the probability that there is no MT satisfying $p_{0j} < P$. 

\textit{Remark: } This theorem has wide application and allows to compute the transmit power of a multi-user transmission where power is minimized given path-loss and shadowing. It is used in this work to analyze the set of served MTs and the interference received at MTs.

\subsection{Analysis of the set of served MTs}
\label{sec:Service_analysis}

\begin{figure*}
\begin{align}
\Lambda_{MT}^{(S)} (0 \; \vert \; l) = 0
\quad \text{and} \quad
\Lambda_{MT}^{(S)} (p \; \vert \; l) = \Lambda_{MT}^{(S)} (p_{l^{\star}-1}^{(\text{cov})} \; \vert \; l) + \frac{\lambda_{MT}}{\lambda_{B}} \frac{1}{\underset{l=l^{\star}}{\overset{\mathsf{L}}{\sum}} v_l}
\left[\exp \left(- \Lambda_{B}^{(A)} (p_{l^{\star}-1}^{(\text{cov})} ) \right)  - \exp \left(- \Lambda_{B}^{(A)} (p) \right) \right]
 \label{eq:lambda_MT_S_conditional} 
\end{align}
\hrule
\end{figure*}

As described in Definition \ref{def:selection_rule}, the typical base station scBS$^{(0)}$ selects MTs in increasing order of required transmit power based on the \emph{effective} set of associated MTs, while cell association is performed by MTs based on an \emph{estimate} set of associated MTs (which depends on the estimate power $\widetilde{\mathsf{P}}_{0\setminus j}^{(T)}$ $\forall j$). As a consequence, there exists a non-zero probability that a MT associated with the typical scBS cannot be effectively served due to power shortage. %
MT$_j$ requiring a transmit power of $p_{0j}$ is rejected by scBS$^{(0)}$ if the current realization of the PPP network and channel gains leads to $ p_{0j} +  \sum_{i \in \mathcal{S}_{0\setminus j}} \lceil p_{0i} \rceil > \mathsf{P}_0^{(A)} $, i.e. if $\widetilde{\mathsf{P}}_{0 \setminus j}^{(T)}$ has been underestimated.
However, the rejection probability is negligible and we propose the following approximation, which is valid only for the proposed cell association.

\begin{approximation}
The subset $\mathcal{S}_0$ of MTs served by the typical small-cell base station scBS$^{(0)}$ is approximated by the subset of MTs associated to it.
\label{approx:Lambda_MT_S}
\end{approximation}
Such approximation upper-bounds the set of served MTs and states that the proposed cell association guarantees service to any associated MT.

\begin{lemma}
Given that $l$ power units are available in the battery, the subset $\mathcal{S}_0$ of MTs served by scBS$^{(0)}$ forms a PPP on $\left[0,p_{l}^{(\text{cov})}\right]$, with:
\begin{align}
d\Lambda_{MT}^{(S)} (p \; \vert \; l) &= d\Lambda_{MT} (p) \mathbb{P}_B^{(\text{Ass})}(p \vert l) \nonumber\\
&= \lambda_{MT} \frac{2}{\alpha} \Upsilon p^{\frac{2}{\alpha}-1} \exp \left(- \Lambda_{B}^{(A)} (p) \right) 
\end{align}
The closed-form expression of $\Lambda_{MT}^{(S)} (p \vert l)$ is computed by integration, as given in Eq. \eqref{eq:lambda_MT_S_conditional} at the top of page.
\end{lemma}

\subsection{Characterization of the battery state probability vector}
\label{sec:Battery_analysis}

As $\Lambda_{MT}^{(S)} (p \; \vert \; l)$ refers to the density of $\mathcal{S}_0$ given that there are $l$ power units in the battery, i.e. the subset of MTs served by scBS$^{(0)}$ at time slot $t$, we apply Theorem \ref{prop:useful_lemma} to compute the total power consumption as follows:
\begin{align}
\mathbb{P}_{\text{T}}(m \; \vert \; l ) = & \left \lbrace
\begin{array}{ll}
 \frac{\mathbb{P}_{\Sigma} \left(m,\Lambda_{MT}^{(S)}(p \; \vert \; l),\lfloor p_{l}^{(\text{cov})} \rfloor\right) }{ \mathbb{P}_{\text{sum}} (l) } & \; \text{if} \; m \leq l \\
0 & \; \text{otherwise}.
\end{array}\right.
\label{eq:prob_T_l}
\\
\text{where } \; \mathbb{P}_{\text{sum}} (l) = & \sum_{m=1}^{l} \mathbb{P}_{\Sigma} \left(m,\Lambda_{MT}^{(S)}(p \; \vert \; l),\lfloor p_{l}^{(\text{cov})} \rfloor\right)
\nonumber 
\end{align}
and $p_{l}^{(\text{cov})}$ refers to the maximum transmit power that can be required by a MT of $\mathcal{S}_0$, as given in Eq. \eqref{eq:p_cov}.
Next, we compute the matrix $\mathbf{P}$. Given the battery model described in Subsection \ref{sec:model_buffer} and depicted in Figure \ref{fig:battery_model}, the probability to go from state $l$ to state $q$  from a time slot to the next one is given by:
\begin{align}
\forall l 
\left \lbrace
\begin{array}{ll}
\mathbb{P}_{\overrightarrow{lq}} =
\sum_m \mathbb{P}_{\text{T}}\left( m \; \vert \; l \right) \mathbb{P}_{\text{H}}\left( q-l+m \right)
& \quad \forall q \neq \mathsf{L} 
\\
\mathbb{P}_{\overrightarrow{l \mathsf{L} }} =
\sum_m \mathbb{P}_{\text{T}}\left( m \; \vert \; l \right) \sum_{q \geq \mathsf{L}-l+m} \mathbb{P}_{\text{H}}\left(q \right)
&
\end{array}
\right.
\label{eq:P_ij}
\end{align}
The case $q=\mathsf{L}$ accounts for the finite battery capacity: no more than $\mathsf{L}$ power units can be stored, whatever the amount of harvested power.

Finally, we propose to numerically solve the steady-state system of Eq. \eqref{eq:steady_state} and obtain $\mathbf{v}$ by applying the algorithm in Table \ref{algo:steady_state}. This algorithm is directly inspired from the fixed-point method commonly used to solve systems of the form $x = f(x)$. 
If $\mathbf{v}^{(0)}$ represents a first guess for the solution, successive approximations to the solution are computed as $\mathbf{v}^{(n+1)} =\mathbf{v}^{(n)} \mathbf{P}$.
Following the Doeblin's theorem for Markov chains, such an algorithm is known to converge to the stationary distribution for all initial distributions.
In our case, it converges after 10 to 20 iterations only.

\begin{table}
\caption{Fixed-point method to obtain  $\mathbf{v} = \left[ v_0 v_1 \ldots v_{\mathsf{L}} \right]$}
\label{algo:steady_state}
\hrule
\begin{algorithmic}[1]
\STATE {$v_l = \frac{1}{\mathsf{L}+1} \quad \forall l$ \hfill \textit{\% Equiprobable states} }
\STATE {$\Delta = 1 > \delta$ \hfill \textit{\% Stopping condition, $\delta = 10^{-10}$ here} }
\STATE {Compute $\mathbb{P}_H(m)$ as in Eq. \eqref{eq:prob_H}}
\WHILE{$\Delta > \delta$}
	\STATE{Compute $\Lambda_{B}^{(A)} (p)$ using Eq. \eqref{eq:Lambda_B_A}}
	\STATE{Compute $\Lambda_{MT}^{(S)} (p \; \vert \; l)$ using Eq. \eqref{eq:lambda_MT_S_conditional}}
	\STATE{Compute $\mathbb{P}_T(m  \; \vert \; l)$ as in Eq. \eqref{eq:prob_T_l}} 
	\STATE{Compute $\mathbf{P} = \left[ \mathbb{P}_{\overrightarrow{lq}}\right]$ as in Eq. \eqref{eq:P_ij}} 
	\STATE {$\Delta = \mathbb{E} \left[ \left( \mathbf{v} - \mathbf{v} \mathbf{P} \right)^2 \right]$ \hfill \textit{\% Mean squared error}}
	\STATE {$\mathbf{v} = \mathbf{v} \mathbf{P}$, ensuring $\sum_l v_l = 1$}
\ENDWHILE
\RETURN {$\mathbf{v}$}
\end{algorithmic}
\hrule
\end{table}

\section{Performance of the proposed association} 
\label{sec:performance_analysis_final}

We now analyze the performance achieved by the proposed energy-aware cell association.

\subsection{Probability of power outage} 
\label{sec:outage_proba}

A power outage occurs either when a MT cannot find any available small-cell base station or when it is associated but dropped due to power shortage. 
With regards to the proposed association, this second event can be neglected. As highlighted in Approximation \ref{approx:Lambda_MT_S}, MTs have very low probability to be dropped once associated.

\begin{proposition}
The probability of power outage $\mathbb{P}_{\text{out}}^{(A)}$ of the proposed cell association is expressed as
\begin{align*}
\mathbb{P}_{\text{out}}^{(A)} = \exp \left(- \Lambda_B^{(A)}\left( p_{\mathsf{L}}^{(\text{cov})} \right) \right) 
\end{align*}
where $\Lambda_B^{(A)}$ is defined in Eq. \eqref{eq:Lambda_B_A} and $p_{\mathsf{L}}^{(\text{cov})}$ is the maximal power that can be required by a MT, solving Eq.\eqref{eq:p_cov}. 
\label{prop:power_outage} 
\end{proposition}
\begin{proof}
First, given that BSs cannot store more than $\mathsf{L}$ power units, the maximum power that can be required by a MT is equal to $ p_{\mathsf{L}}^{(\text{cov})}$, regardless of the current BS battery state. Thus, any BS$_k$ for which $p_{k0} > p_{\mathsf{L}}^{(\text{cov})}$ is declared non-available. 
Second, not all BSs for which the required transmit power satisfies $p_{k0} \leq p_{\mathsf{L}}^{(\text{cov})}$ is available, since they may lack of power. Accounting for the random battery fluctuations, their density is reduced from $\Lambda_B$ to $\Lambda_B^{(A)}$. Finally, using the void probability of PPPs \cite{book_sto_geo}, we deduce Proposition \ref{prop:power_outage}.
\end{proof}

\subsection{Coverage probability} 
\label{sec:SINR}

The coverage probability $\mathbb{P}_\gamma$ is defined as the probability that the signal-to-interference ratio (SIR) at MT$^{(0)}$ exceeds a given reliability threshold $T$. While the availability criteria and power outage probability only depend on path-loss and shadowing, this performance metric accounts for fast fading and interference as well.

We denote as $N_{RB}$ the number of resource blocks available for data transmission at each small-cell BS. They are assumed picked at random, each with equal probability $\frac{1}{N_{RB}}$. First, the expected transmit power of scBS$_k, \forall k$ in the resource block used by MT$^{(0)}$ is equal to $\frac{\mathsf{P}_k^{(T)} }{ N_{RB}}$,  with $\mathsf{P}_k^{(T)}$ the total transmit power of scBS$_k$. Denoting scBS$_0$ the small base station serving MT$^{(0)}$, the interference $I$ received at MT$^{(0)}$ is computed as:
\begin{align}
I & = \underset{\text{scBS}_k \in \Phi_B \setminus \text{scBS}_0 }{\sum} \frac{\mathsf{P}_k^{(T)} }{ N_{RB}} \frac{\chi_{k,0}}{l(r_{k,0})}\vert h_{k,0} \vert^2
\\
& = \sum_{m=1}^{\mathsf{L}} \sum_{l=m}^{\mathsf{L}} {\sum}_{\Phi_{m,l}^{(I)}} \frac{m}{ N_{RB}} \frac{\chi_{k,0}}{l(r_{k,0})}\vert h_{k,0} \vert^2,
\end{align}
where $m$ refers a total transmit power, $l$ to a BS battery state and $\Phi_{m,l}^{(I)}$ refers to the set of interfering scBSs having $l$ units of energy stored in their battery and transmitting with $m \leq l$ power units. Note that the $\Phi_{m,l}^{(I)}$ form a partition of the interfering scBSs, as $\Phi_B \setminus \text{scBS}_0  = \underset{m=0}{\overset{\mathsf{L}}{\bigcup}} \; \underset{l=m}{\overset{\mathsf{L}}{\bigcup}} \Phi_{m,l}^{(I)}$.

\begin{proposition}
The coverage probability of the proposed cell association is given by:
\begin{align*}
& \mathbb{P}_\gamma^{(A)}  = 
\prod_{m=1}^{\mathsf{L}}   \prod_{l=m}^{\mathsf{L}} \exp \left( - 
\rho_{m,l} \lambda_{B} \Upsilon_m \frac{2}{\alpha}
\left(\frac{T}{\mathsf{P_{Rx}} }\right)^{\frac{2}{\alpha}}
\int_{u_{m,l}}^{\infty} \frac{u^{2/\alpha -1}}{1+u}du
\right)
\\
& \text{where} \quad \left \lbrace \begin{array}{l l}
\Upsilon_m &= \pi \left(\frac{m}{\kappa N_{RB} }\right)^{\frac{2}{\alpha}} 
\exp \left(  \frac{2/\alpha}{\zeta} \mu + \frac{1}{2} \left(\frac{2/\alpha}{\zeta}\right)^2 \sigma^2 \right) 
\\
\rho_{m,l} &= v_l \mathbb{P}_{\text{T}}(m \; \vert \; l) 
\\
u_{m,l} &= \min \left(\frac{1}{T} , \frac{p_l^{(\text{cov})} N_{RB}}{m T}  \right)
\end{array} \right. 
\end{align*}
In this, $\rho_{m,l} \lambda_{B}$ refers to the density of $\Phi_{m,l}^{(I)}$ and $\mathbb{P}_{\text{T}}(m \; \vert \; l) $ is given by Eq. \eqref{eq:prob_T_l}, using Theorem \ref{prop:useful_lemma}.
\label{prop:coverage_probability}
\end{proposition}

\begin{proof}
See Appendix \ref{appendix:coverage_prob}.
\end{proof}
Such expression of the coverage probability is derived from the Laplace transform of the received interference in PPP networks. This stochastic-geometry tool has been widely used in conventional cellular networks with fixed transmit power, as in \cite{Andrews2011_tractable}.
In our case, the BS transmit power is not constant but partitioning the BSs into the sets  $\Phi_{m,l}^{(I)}$ allows to consider a fixed average transmit power equal to $\frac{m}{ N_{RB}}$ for analyzing the received interference.
Another difference with existing frameworks dedicated to non energy-harvesting networks lies in the power received from the closest interfering scBS, measured by $\frac{\mathsf{P_{Rx}}}{u_{m,l} T}$. Whereas an interferer cannot be closer than $r_0$ (users associate with the strongest BS), it is not the case in our framework due to the availability criteria. Indeed, a nearby BS may be declared unavailable for a given user, which then has to associate with a farther BS.

\textit{Remark:} Theorem \ref{prop:useful_lemma} allows to discard the assumption of full-power transmission with saturated traffic load in the computation of the received interference, as commonly used for tractable analysis \cite{Dhillon2014_Fundamentals, Yu2015}. With our scenario and considered network model, such assumption would imply that the transmit power model for the battery discharge (which accounts for the potentially low number of served MTs) does not match the transmit power model for the received interference (which gives a worst-case result for the coverage probability). In the proposed analysis, the effective inter-cell interference is computed based on the battery level and power requirements of associated MTs. Thus, a scBS with low battery level and serving few MTs also generates less interference.

\section{Simulations and Performance results}  
\label{sec:simulations}

We first validate the proposed analysis and then, present the performance results obtained for the power-availability-aware cell association, with a particular focus on the power outage. %

\subsection{Considered reference cell association strategies}
\label{sec:sim_setup_ref}

In the following, the proposed cell association is referred as "$\mathcal{A}$".
As first reference,
 we consider the performance achieved when scBSs are connected to the power grid and do not require energy harvesting. As for the proposed strategy, we assume that the transmit power is minimized and that users associate with the scBS consuming the least power. A maximum transmit power constraint $\mathsf{P_{max}^{(OG)}}$ is assumed for data transmission towards each user. Such scheme is referred as "\textit{on-grid}" and gives performance upper bound. %

As second reference scheme,
we consider a more conventional association policy with energy harvesting but without power availability checking. Similarly to $\mathcal{A}$, a MT associates to the scBS for which the required transmit power is minimal, as in Definition \ref{def:closest_power}, but the battery level is not considered and no pre-selection of available base stations is performed. In this case, $\mathcal{A}_j = \Phi_{B}$ for all $j$. This scheme is referred as  "\textit{w/o}$\mathcal{A}$".

As third reference,
we consider a cell association policy similar to the proposed one, with power availability criteria and minimal-power association, but where the BSI (battery state information) can be known in real-time, at any instant. In this case, a MT performs cell association knowing the effective instantaneous battery level of each small-cell base station. 
The power availability criteria is thus reduced to $p_{kj} \leq \mathsf{P}_k^{(A)}$ and, contrary to the proposed strategy, MTs are served in a first-come first-served basis and not in ascending order of required power, as no non-causal information is assumed. Such scheme is denoted as "\textit{rt-}$\mathcal{A}$".

\textit{Remark:} Performance upper-bounds can be obtained by considering centralized cell association, where the scBS-MT pairing is computed at each instant with global instantaneous knowledge of channel and battery states. Yet, such optimized algorithms require large amount of overhead signaling. The resulting extra power consumption may rapidly depletes the BS battery but is generally omitted for performance analysis. Computing the optimal solution within the considered framework and proposing an optimal/suboptimal algorithm for it largely exceeds the scope of this paper.

\begin{table}
	\renewcommand{\arraystretch}{1.3}
	\caption{Default simulation parameters}
	\label{table:simulation}
	\centering
	\begin{tabular}{|c|c||c|c||c|c|}
		\hline 
		\multicolumn{2}{|c||}{Channel} &
		\multicolumn{2}{|c||}{scBS} &
		\multicolumn{2}{|c|}{MT} \\[3pt]
		\hline
		$\kappa$ / $\alpha$ & 1 /4 &
		$\lambda_{BS}$ & $\frac{1}{\pi 60^2}$ &
		$\lambda_{MT}$ & $\frac{15}{\pi 60^2}$ \\[3pt]
		
		$\mu$ / $\sigma$ & 0 / 4dB &
		$\mathsf{L}$ / $\mathsf{P}_{\max}$ & 1000 / 1W &
		$\mathsf{P_{Rx}}$ & -65dBm  \\[3pt]
		
		$\nu$ & 1 &
		$\lambda_e$ / $N_e$ & 10\%$\mathsf{L}$ / 1 &
		$\mathsf{P_{max}^{(OG)}}$ & 50mW  \\[3pt]
		\hline
	\end{tabular}
\end{table}

 \subsection{Simulation setup and Analysis validation}

To simulate the network performance and validate the proposed analysis, we use Monte-Carlo simulations. Each simulation trial first consists in generating the base station PPP in a finite window. The battery level at each BS is initially set as random, taken from a uniform distribution. Then, a sufficiently large number %
 of sub-trials is simulated to model the BSI broadcast periods.
Each of them consists in generating the energy arrivals at each scBS and the user PPP in the considered window. Next, cell association is performed, the battery level is updated following the effective harvested and consumed energy and both outage and coverage are evaluated for this sub-trial. Note that, to avoid any result distortion due to the initial conditions, a transition period is simulated at the beginning of each trial such that each battery reaches its stationary state. 
Finally, the outage and coverage probabilities are computed by averaging over 10,000 such trials. If not specified, we consider the simulation parameters of Table \ref{table:simulation}.

Before moving onto performance results, we validate the analysis proposed in this paper. The outage and coverage probabilities obtained by Monte-Carlo simulations are depicted in Figures \ref{fig:outage_Prob_BS}, \ref{fig:outage_PRx} and \ref{fig:coverage_Prob} for a wide range of simulation parameters.
The plots show a complete agreement between the analytical and simulated results. In particular, Figures \ref{fig:outage_Prob_BS} and \ref{fig:outage_PRx} depict both the simulated probability that MT$_0$ cannot associate with any scBS (namely, "simul. 1-$\mathbb{P}_B^{(\text{Ass})}	$" in the Figure) and the effective simulated power outage probability (namely, "simul. $\mathbb{P}_{\text{out}}^{(A)}$"), which also includes the cases when MT$_0$ is associated with a scBS but not selected for service. The good match of both of them with the power outage as computed in Proposition \ref{prop:power_outage} validates the proposed analysis and Approximation \ref{approx:Lambda_MT_S}.

\subsection{Performance gain reached by the proposed cell association}
\label{sec:sim_results}

\subsubsection{Robustness to bursty energy arrivals}
Along with its uncertainty, the energy arrival process is characterized by its irregularity, which may have a highly detrimental effect on the network performance.
As first part of the simulation results, we evaluate the robustness of the proposed association against bursty arrivals. To do so, we consider that energy is harvested at rate $\frac{\lambda_e}{N_e}$ over one time slot and that each arrival consists of $N_e \tau$ energy units. When $N_e$ increases, energy is harvested less frequently, but each burst carries more power units.

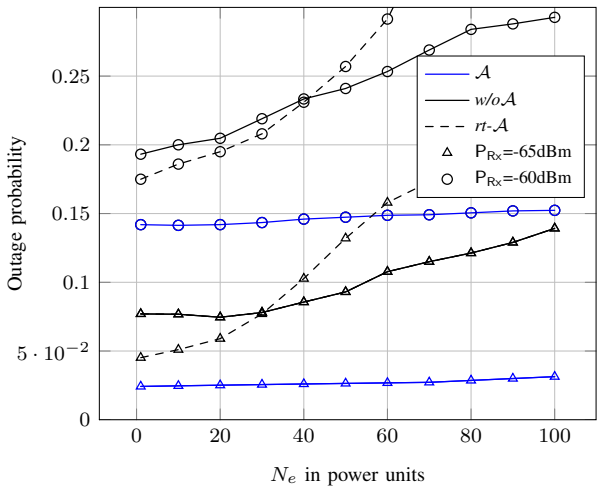
\begin{figure}
\centering
\centering \resizebox{0.9\columnwidth}{!}{%
\begin{tikzpicture}[scale=4,cap=round,>=latex]

  \begin{axis}[%
  ymin=0, ymax = 0.3,
  ytick= {0, 0.05 , ..., 0.3},
	xlabel= $N_e$ in power units,%
    ylabel= Outage probability,%
    y label style={at={(axis description cs:0.05,0.5)},rotate=0,anchor=south},
	every axis/.append style={font=\footnotesize},  
    grid=major,%
    legend style={nodes=right,font=\scriptsize,anchor=north west, at={(axis cs: 67,0.265)}},%
    mark size=2pt]

\addplot[solid,color=blue, line width=0.5] table[x=NEH ,y=outage ,col sep=semicolon] {./Prob_Outage_Proposed_forNEH_6510.txt};
\addlegendentry{$\mathcal{A}$};

\addplot[solid,color=black, line width=0.5] table[x=NEH ,y=outage ,col sep=semicolon] {./Prob_Outage_Without_forNEH_6510.txt};
\addlegendentry{\textit{w/o}$\mathcal{A}$};

\addplot[dashed,color=black, line width=0.5] table[x=NEH ,y=outage ,col sep=semicolon] {./Prob_Outage_Without_forNEH_6510.txt};
\addlegendentry{\textit{rt-}$\mathcal{A}$};

\addplot[only marks,color=black, line width=0.5, mark=triangle,mark options={solid}, line width=0.5] table[x=NEH ,y=outage ,col sep=semicolon] {./Prob_Outage_Without_forNEH_6510.txt};
\addlegendentry{$\mathsf{P_{Rx}}$=-65dBm};

\addplot[only marks,color=black, line width=0.5, mark=o,mark options={solid}, line width=0.5] table[x=NEH ,y=outage ,col sep=semicolon] {./Prob_Outage_Proposed_forNEH_6010.txt};
\addlegendentry{$\mathsf{P_{Rx}}$=-60dBm};

\addplot[solid,color=blue, mark size =2, mark=o,mark options={solid},line width=0.5] table[x=NEH ,y=outage ,col sep=semicolon] {./Prob_Outage_Proposed_forNEH_6010.txt};

\addplot[solid,color=blue, mark size =2, mark=triangle,mark options={solid}, line width=0.5] table[x=NEH ,y=outage ,col sep=semicolon] {./Prob_Outage_Proposed_forNEH_6510.txt};

\addplot[solid,color=black, mark size =2, mark=o,mark options={solid}, mark phase=1,line width=0.5] table[x=NEH ,y=outage ,col sep=semicolon] {./Prob_Outage_Without_forNEH_6010.txt};

\addplot[solid,color=black, mark size =2, mark=triangle,mark options={solid}, mark phase=1,line width=0.5] table[x=NEH ,y=outage ,col sep=semicolon] {./Prob_Outage_Without_forNEH_6510.txt};

\addplot[dashed,color=black, mark size =2, mark=o,mark options={solid}, mark phase=1,line width=0.5] table[x=NEH ,y=outage ,col sep=semicolon] {./Prob_Outage_RealTime_forNEH_6010.txt};

\addplot[dashed,color=black, mark size =2, mark=triangle,mark options={solid}, mark phase=1,line width=0.5] table[x=NEH ,y=outage ,col sep=semicolon] {./Prob_Outage_RealTime_forNEH_6510.txt};

    \end{axis}
\end{tikzpicture}
}
 \caption{Impact of the burstiness of energy arrivals on the power outage}
\label{fig:outage_NEH}
\end{figure}

The outage probabilities obtained for schemes $\mathcal{A}$, \textit{w/o-}$\mathcal{A}$ and \textit{rt-}$\mathcal{A}$ are illustrated in Figure \ref{fig:outage_NEH}.
Significant loss in the power outage is observed for both \textit{w/o-}$\mathcal{A}$ and \textit{rt-}$\mathcal{A}$ when the energy harvesting process is bursty (high values for $N_e$).
As no energy management is performed and users associate with scBSs ignoring their battery level, the reference \textit{w/o-}$\mathcal{A}$ is severely handicapped by long periods without energy arrivals. In addition, more energy is generally wasted at scBSs since energy may be harvested at a scBS and few users may effectively associate to it.
Next, the performance of \textit{rt-}$\mathcal{A}$ is essentially conditioned by the arrival order of the power bursts and MTs service requests, given the considered first-come first-served approach. Receiving large power bursts allows serving users at cell edge or with weak channel, i.e. users with high power requirement, but this rapidly depletes the battery and prevents from serving other users at next time slot.

On the contrary, the proposed power availability criteria is only slightly affected by the burstiness of the energy arrival process. Such gain is further illustrated in Figures \ref{fig:illustration_association_voronoi}(b) and (c).
As traffic is offloaded from scBSs with low battery to scBSs with higher battery and users are served in ascending order of power requirement, the overall available energy is more evenly distributed among base stations to serve users. 
As a consequence and following principles enunciated in \cite{Huang2015}, the proposed availability checking can be considered as a non-direct energy-sharing mechanism within the wireless communication networks, in the same way as node cooperation.

\textit{Remark:} As observed in Figure \ref{fig:outage_NEH}, the case $N_e =1$ (regular energy arrivals) presents the minimal gain of the proposed cell association over both \textit{w/o-}$\mathcal{A}$ and \textit{rt-}$\mathcal{A}$. As it constitutes a lower-bound on the achievable performance gain, it is the only case considered in the following.

\begin{figure}
\centering
\centering \resizebox{0.85\columnwidth}{!}{%
\begin{tikzpicture}[scale=4,cap=round,>=latex]

  \begin{axis}[%
  ymin=-0.01, ymax = 0.68,
	xlabel= {Average scBS coverage radius $\lambda_B =\frac{1}{\pi R^2}$ (in m)},%
    ylabel= Outage probability,%
    y label style={at={(axis description cs:0.1,0.5)},rotate=0,anchor=south},
	every axis/.append style={font=\footnotesize},  
    grid=major,%
    legend style={nodes=right,anchor= north west, font=\scriptsize, at={(axis cs:45,0.68)}},%
    mark size=2pt]

\addplot[solid,color=blue, line width=0.5] table[x=BS ,y=outage ,col sep=semicolon] {./Prob_Outage_Proposed_forBS_6510.txt};
\addlegendentry{$\mathcal{A}$};

\addplot[solid,color=black, line width=0.5] table[x=BS ,y=outage ,col sep=semicolon] {./Prob_Outage_Without_forBS_6510.txt};
\addlegendentry{\textit{w/o}$\mathcal{A}$};

\addplot[dashed,color=black, line width=0.5] table[x=BS ,y=outage ,col sep=semicolon] {./Prob_Outage_RealTime_forBS_6510.txt};
\addlegendentry{\textit{rt-}$\mathcal{A}$};

\addplot[dotted,color=black,line width=0.5] table[x=BS ,y=outage_50 ,col sep=semicolon] {./Prob_Outage_OG_forBS_65.txt};
\addlegendentry{\textit{on-grid} - 50mW};

\addplot[only marks,color=black, mark size =2, line width=0.5, mark=triangle,mark options={solid}, line width=0.5] table[x=BS ,y=outage ,col sep=semicolon] {./Prob_Outage_Without_forBS_6510.txt};
\addlegendentry{$\mathsf{P_{Rx}}$=-65dBm};

\addplot[only marks,color=black, line width=0.5, mark=o,mark options={solid}, line width=0.5] table[x=BS ,y=outage ,col sep=semicolon] {./Prob_Outage_Proposed_forBS_6510.txt};
\addlegendentry{$\mathsf{P_{Rx}}$=-60dBm};

\addplot[dashed,color=green!80!black, mark size =3, mark=otimes,mark options={solid}, line width=0.5] table[x=BS ,y=outageTrueAvailable  ,col sep=semicolon] {./Prob_Outage_Proposed_forBS_6510.txt};
\addlegendentry{simul. 1-$\mathbb{P}_B^{(\text{Ass})}	$};

\addplot[dashed,color=red!80, mark size =3, mark=square,mark options={solid}, line width=0.5] table[x=BS ,y=outageTrue ,col sep=semicolon] {./Prob_Outage_Proposed_forBS_6510.txt};
\addlegendentry{simul. $\mathbb{P}_{\text{out}}^{(A)}$};

\addplot[solid,color=blue, mark size =2, mark=triangle,mark options={solid}, line width=0.5] table[x=BS ,y=outage ,col sep=semicolon] {./Prob_Outage_Proposed_forBS_6510.txt};

\addplot[solid,color=black, mark size =2, mark=triangle,mark options={solid}, line width=0.5] table[x=BS ,y=outage ,col sep=semicolon] {./Prob_Outage_Without_forBS_6510.txt};

\addplot[dashed,color=black, mark size =2, mark=triangle,mark options={solid}, line width=0.5] table[x=BS ,y=outage ,col sep=semicolon] {./Prob_Outage_RealTime_forBS_6510.txt};

\addplot[dotted,color=black, mark size =2, mark=triangle,mark options={solid}, line width=0.5] table[x=BS ,y=outage_50 ,col sep=semicolon] {./Prob_Outage_OG_forBS_65.txt};

\addplot[solid,color=blue, mark size =2, mark=o,mark options={solid}, line width=0.5] table[x=BS ,y=outage ,col sep=semicolon] {./Prob_Outage_Proposed_forBS_6010.txt};

\addplot[solid,color=black, mark size =2, mark=o,mark options={solid}, line width=0.5] table[x=BS ,y=outage ,col sep=semicolon] {./Prob_Outage_Without_forBS_6010.txt};

\addplot[dashed,color=black, mark size =2, mark=o,mark options={solid}, line width=0.5] table[x=BS ,y=outage ,col sep=semicolon] {./Prob_Outage_RealTime_forBS_6010.txt};

\addplot[dotted,color=black, mark size =2, mark=o,mark options={solid}, line width=0.5] table[x=BS ,y=outage_50 ,col sep=semicolon] {./Prob_Outage_OG_forBS_60.txt};

\draw[red] (axis cs: 90, 0.4) ellipse (5pt and 15pt);
\draw[red] (axis cs: 85, 0.5) ellipse (5pt and 14pt);

    \end{axis}
\end{tikzpicture}
} \caption{Impact of the scBS density on the power outage}
\label{fig:outage_Prob_BS}
\end{figure}
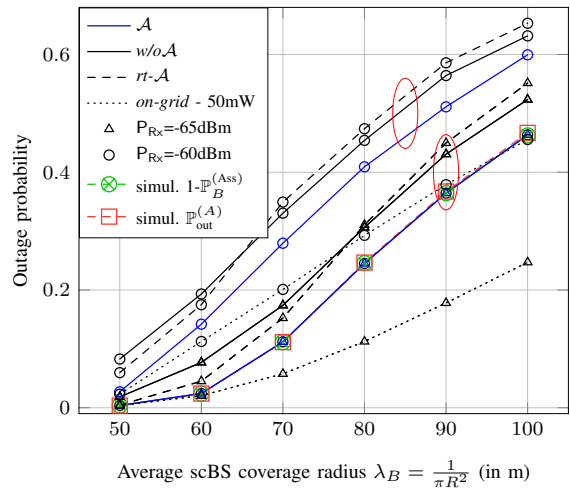

\subsubsection{Impact of the BS density}

Next, we analyze the impact of the scBS density on the probability of power outage, as depicted in Figure \ref{fig:outage_Prob_BS}. Comparing the proposed cell association with \textit{w/o-}$\mathcal{A}$, we observe that the power availability checking provides a relatively constant gain, which increases when the received power constraint $\mathsf{P_{Rx}}$ decreases. Indeed,  with smaller $\mathsf{P_{Rx}}$, user requires less transmit power and can generally associate with more base stations, even if located farther. Hence, a smoother distribution of the available energy can be reached among the scBSs in the network, by non-direct energy sharing. By comparison with the \textit{on-grid} strategy, the outage loss due to energy harvesting is deeper for large cells. Indeed, the considered association strategies with energy-harvesting base stations are limited in their total power consumption by the battery and cannot serve far users due to power shortage, while we assume for the \textit{on-grid} strategy a transmit power constraint at each user only, and not on the overall consumption.

\subsubsection{Received power constraint and harvesting rate}

The probability of power outage is then plotted as a function of the received power constraint $\mathsf{P_{Rx}}$ in Figure \ref{fig:outage_PRx}. We observe that the proposed energy-aware outperforms both \textit{w/o-}$\mathcal{A}$ and \textit{rt-}$\mathcal{A}$ and approaches the performance reached by the \textit{on-grid} policy for $\mathsf{P_{max}^{(OG)}}=$50mW.
Furthermore, the obtained gain increases when the harvesting rate $\lambda_e$ diminishes. This suggests that the power availability checking presented in Definition \ref{def:available} brings robustness to impaired harvesting conditions, i.e. cloudy days for solar photovoltaic panels.
Note that such simulation results have been presented for $N_e=1$, which lower-bounds the gain of the proposed strategy. 

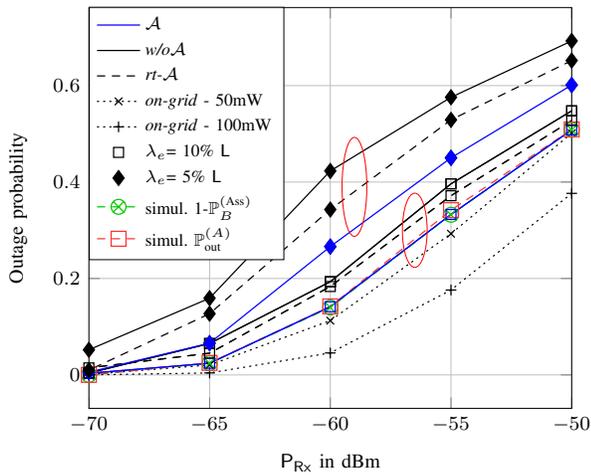
\begin{figure}
\centering
\centering \resizebox{0.9\columnwidth}{!}{%
\begin{tikzpicture}[scale=4,cap=round,>=latex]

  \begin{axis}[%
  xmin=-70, xmax = -50,
	xlabel= $\mathsf{P_{Rx}}$ in dBm,%
    ylabel= Outage probability,%
    y label style={at={(axis description cs:0.07,0.5)},rotate=0,anchor=south},
	every axis/.append style={font=\footnotesize},  
    grid=major,%
    legend style={nodes=right,font=\scriptsize, at={(axis cs:-70,0.238)}, anchor=south west},%
    mark size=2pt]

\addplot[solid,color=blue, line width=0.5] table[x=PRx ,y=outage ,col sep=semicolon] {./Prob_Outage_Proposed_forPrx_6010.txt};
\addlegendentry{$\mathcal{A}$};

\addplot[solid,color=black, line width=0.5] table[x=PRx ,y=outage ,col sep=semicolon] {./Prob_Outage_Without_forPrx_6010.txt};
\addlegendentry{\textit{w/o}$\mathcal{A}$};

\addplot[dashed,color=black, line width=0.5] table[x=PRx ,y=outage ,col sep=semicolon] {./Prob_Outage_RealTime_forPrx_6010.txt};
\addlegendentry{\textit{rt-}$\mathcal{A}$};

\addplot[dotted,color=black, mark size =2, mark=x,mark options={solid}, line width=0.5] table[x=PRx ,y=outage_50 ,col sep=semicolon] {./Prob_Outage_OG_forPrx_60.txt};
\addlegendentry{\textit{on-grid} - 50mW};

\addplot[dotted,color=black, mark size =2, mark=+,mark options={solid}, line width=0.5] table[x=PRx ,y=outage_100 ,col sep=semicolon] {./Prob_Outage_OG_forPrx_60.txt};
\addlegendentry{\textit{on-grid} - 100mW};

\addplot[only marks,color=black, line width=0.5, mark=square,mark options={solid}, line width=0.5] table[x=PRx ,y=outage ,col sep=semicolon] {./Prob_Outage_Without_forPrx_6010.txt};
\addlegendentry{$\lambda_e$= 10\% $\mathsf{L}$};

\addplot[only marks,color=black, mark size =3, line width=0.5, mark=diamond*,mark options={solid}, line width=0.5] table[x=PRx ,y=outage ,col sep=semicolon] {./Prob_Outage_Proposed_forPrx_605.txt};
\addlegendentry{$\lambda_e$= 5\% $\mathsf{L}$};

\addplot[dashed,color=green!80!black, mark size =3, mark=otimes,mark options={solid}, line width=0.5] table[x=PRx ,y=outageTrueAvailable  ,col sep=semicolon] {./Prob_Outage_Proposed_forPrx_6010.txt};
\addlegendentry{simul. 1-$\mathbb{P}_B^{(\text{Ass})}	$};

\addplot[dashed,color=red!80, mark size =3, mark=square,mark options={solid}, line width=0.5] table[x=PRx ,y=outageTrue ,col sep=semicolon] {./Prob_Outage_Proposed_forPrx_6010.txt};
\addlegendentry{simul. $\mathbb{P}_{\text{out}}^{(A)}$};

\addplot[solid,color=blue, mark size =2, mark=square,mark options={solid}, line width=0.5] table[x=PRx ,y=outage ,col sep=semicolon] {./Prob_Outage_Proposed_forPrx_6010.txt};

\addplot[dashed,color=black, mark size =2, mark=square,mark options={solid}, line width=0.5] table[x=PRx ,y=outage ,col sep=semicolon] {./Prob_Outage_RealTime_forPrx_6010.txt};

\addplot[solid,color=black, mark size =2, mark=square,mark options={solid}, line width=0.5] table[x=PRx ,y=outage ,col sep=semicolon] {./Prob_Outage_Without_forPrx_6010.txt};

\addplot[solid,color=blue, mark size =3, mark=diamond*,mark options={solid}, line width=0.5] table[x=PRx ,y=outage ,col sep=semicolon] {./Prob_Outage_Proposed_forPrx_605.txt};

\addplot[dashed,color=black, mark size =3, mark=diamond*,mark options={solid}, line width=0.5] table[x=PRx ,y=outage ,col sep=semicolon] {./Prob_Outage_RealTime_forPrx_605.txt};

\addplot[solid,color=black, mark size =3, mark=diamond*,mark options={solid}, line width=0.5] table[x=PRx ,y=outage ,col sep=semicolon] {./Prob_Outage_Without_forPrx_605.txt};

\draw[red] (axis cs: -59, 0.39) ellipse (5pt and 20pt);
\draw[red] (axis cs: -56.5, 0.3) ellipse (5pt and 15pt);

    \end{axis}
\end{tikzpicture}
}

 \caption{Impact of the required received power constraint $\mathsf{P_{Rx}}$}
\label{fig:outage_PRx}
\end{figure}

\subsubsection{On the frequency of battery broadcast}

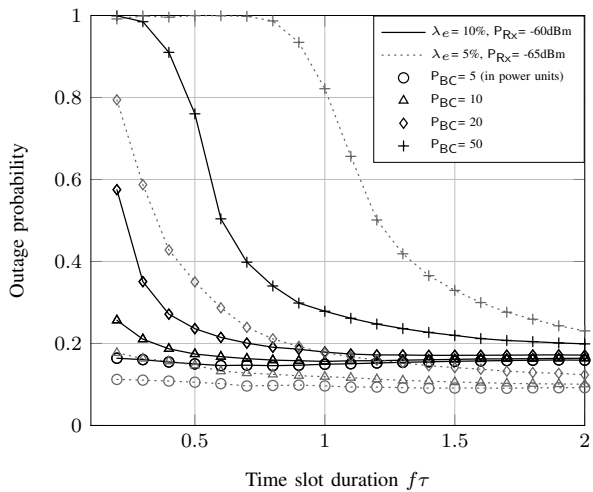
\begin{figure}
	\centering
\centering \resizebox{0.9\columnwidth}{!}{%
\begin{tikzpicture}[scale=4,cap=round,>=latex]

  \begin{axis}[%
  xmin=0.1, xmax = 2,
  ymin=0, ymax = 1,
	xlabel= Time slot duration $f \tau$,%
    ylabel= Outage probability,%
    y label style={at={(axis description cs:0.07,0.5)},rotate=0,anchor=south},
	every axis/.append style={font=\footnotesize},  
    grid=major,%
    legend style={nodes=right,font=\tiny, at={(axis cs:2,1)}, anchor=north east},%
    mark size=2pt]

\addplot[solid,color=black, line width=0.5] table[x=TS ,y=p50 ,col sep=semicolon] {./Prob_Outage_Proposed_freqBC_6010.txt};
\addlegendentry{$\lambda_e$= 10\%, $\mathsf{P_{Rx}}$= -60dBm};

\addplot[dotted,color=black!60, line width=0.5] table[x=TS ,y=p50 ,col sep=semicolon] {./Prob_Outage_Proposed_freqBC_655.txt};
\addlegendentry{$\lambda_e$= 5\%, $\mathsf{P_{Rx}}$= -65dBm};

\addplot[only marks,color=black, mark size =2, mark=o,mark options={solid}, line width=0.5] table[x=TS ,y=p5 ,col sep=semicolon] {./Prob_Outage_Proposed_freqBC_6010.txt};
\addlegendentry{$\mathsf{P_{BC}}$= 5 (in power units)};

\addplot[only marks,color=black, mark size =2, mark=triangle,mark options={solid}, line width=0.5] table[x=TS ,y=p10 ,col sep=semicolon] {./Prob_Outage_Proposed_freqBC_6010.txt};
\addlegendentry{$\mathsf{P_{BC}}$= 10};

\addplot[only marks,color=black, mark size =2, mark=diamond,mark options={solid}, line width=0.5] table[x=TS ,y=p20 ,col sep=semicolon] {./Prob_Outage_Proposed_freqBC_6010.txt};
\addlegendentry{$\mathsf{P_{BC}}$= 20};

\addplot[only marks,color=black, mark size =2, mark=+,mark options={solid}, line width=0.5] table[x=TS ,y=p50 ,col sep=semicolon] {./Prob_Outage_Proposed_freqBC_6010.txt};
\addlegendentry{$\mathsf{P_{BC}}$= 50};

\addplot[solid,color=black, mark size =2, mark=o,mark options={solid}, line width=0.5] table[x=TS ,y=p5 ,col sep=semicolon] {./Prob_Outage_Proposed_freqBC_6010.txt};

\addplot[solid,color=black, mark size =2, mark=triangle,mark options={solid}, line width=0.5] table[x=TS ,y=p10 ,col sep=semicolon] {./Prob_Outage_Proposed_freqBC_6010.txt};

\addplot[solid,color=black, mark size =2, mark=diamond,mark options={solid}, line width=0.5] table[x=TS ,y=p20 ,col sep=semicolon] {./Prob_Outage_Proposed_freqBC_6010.txt};

\addplot[dotted,color=black!60, mark size =2, mark=o,mark options={solid}, line width=0.5] table[x=TS ,y=p5 ,col sep=semicolon] {./Prob_Outage_Proposed_freqBC_655.txt};

\addplot[dotted,color=black!60, mark size =2, mark=triangle,mark options={solid}, line width=0.5] table[x=TS ,y=p10 ,col sep=semicolon] {./Prob_Outage_Proposed_freqBC_655.txt};

\addplot[dotted,color=black!60, mark size =2, mark=diamond,mark options={solid}, line width=0.5] table[x=TS ,y=p20 ,col sep=semicolon] {./Prob_Outage_Proposed_freqBC_655.txt};

\addplot[only marks,color=black!60, mark size =2, mark=+,mark options={solid}, line width=0.5] table[x=TS ,y=p50 ,col sep=semicolon] {./Prob_Outage_Proposed_freqBC_655.txt};

    \end{axis}
\end{tikzpicture}
}

 	\caption{BS battery broadcast: dissipated power and frequency of broadcast}
	\label{fig:freq_BC}
\end{figure}

So far, we have omitted the power $\mathsf{P_{BC}}$ that is consumed by each scBS to periodically broadcast its battery level. To analyze its impact on the power outage probability, we now assume that the time slot duration $\tau$ is scaled by a factor $f \in [0.1,2]$. We adapt the algorithm proposed in Table \ref{algo:steady_state} to this new assumption by limiting the battery capacity to $\mathsf{P}_{\max} - \mathsf{P_{BC}} $ and setting the density of users (resp. of power arrivals) to $\lambda_{MT} f$ (resp. $\lambda_e f$), following the properties of Poisson processes.

As observed in Figure \ref{fig:freq_BC}, the power $\mathsf{P_{BC}}$ dissipated at each time slot significantly impairs the outage probability when the time slot duration is short ($f \leq 0.75$). Indeed, the harvested power units are mostly consumed for battery broadcasting, letting little power available for data transmission.
However, and quite counter-intuitively, the outage probability tends to a fixed value when $f$ increases. The impact of $\mathsf{P_{BC}}$ is then negligible and the performance is solely affected by the ratio between $\lambda_e f $ and $\lambda_{MT} f$, which is constant in our case.
This implies that (i) the battery broadcasts can be spaced out as long as the Poisson assumption for power unit arrivals and number of MTs is valid over the time slot and (ii) that the time slot duration should be set as a function of the traffic variations and harvesting conditions. Stable traffic demand and slowly varying weather conditions (in case of solar panels) allows significant reduction in the required frequency of BS battery broadcasting.

\subsubsection{On the coverage probability}

\begin{figure}
	\centering
\centering \resizebox{0.9\columnwidth}{!}{%
\begin{tikzpicture}[scale=4,cap=round,>=latex]

  \begin{axis}[%
  xmin=-5, xmax = 15,
  ymin=0.9, ymax = 1.0001,
	xlabel= SINR threshold (in dBm),%
    ylabel= Coverage probability,%
    y label style={at={(axis description cs:0.07,0.5)},rotate=0,anchor=south},
	every axis/.append style={font=\footnotesize},  
    grid=major,%
    legend style={nodes=right,font=\scriptsize},%
    legend pos={south west},%
    mark size=2pt]

\addplot[solid,color=blue, ,line width=0.5] table[x=Threshold ,y=NEH_80 ,col sep=semicolon] {./Prob_Coverage_Prop_60m_65dB_1080.txt};
\addlegendentry{$\mathcal{A}$};

\addplot[dashed,color=red!80, mark size =2.5, mark=square,mark options={solid}, mark repeat=2,mark phase=1, line width=0.5] table[x=Threshold ,y=NEH_80 ,col sep=semicolon] {./Prob_Coverage_Prop_60m_65dB_1080_compute.txt};
\addlegendentry{simul. $\mathbb{P}_\gamma^{(A)}$};

\addplot[solid,color=black, line width=0.5] table[x=Threshold ,y=NEH_80 ,col sep=semicolon] {./Prob_Coverage_Without_60m_65dB_1080.txt};
\addlegendentry{\textit{w/o-}$\mathcal{A} \simeq \mathcal{A}$ };

\addplot[dashed,color=black, line width=0.5] table[x=Threshold ,y=NEH_80 ,col sep=semicolon] {./Prob_Coverage_RealTime_60m_65dB_1080.txt};
\addlegendentry{\textit{rt-}$\mathcal{A}$};

\addplot[dotted,color=black, line width=0.5] table[x=Threshold ,y=Set1 ,col sep=semicolon] {./Prob_Coverage_OG_60_0050.txt};
\addlegendentry{\textit{on-grid} - 50mW};

\addplot[only marks,color=black, mark size =2.5, mark=triangle*,mark options={solid}, mark repeat=3,mark phase=1,line width=0.5] table[x=Threshold ,y=NEH_80 ,col sep=semicolon] {./Prob_Coverage_RealTime_60m_65dB_1080.txt};
\addlegendentry{-65dBm / 80 / 60m};

\addplot[only marks,color=black, mark size =2, mark=triangle,mark options={solid}, mark repeat=3,mark phase=2,line width=0.5] table[x=Threshold ,y=NEH_1 ,col sep=semicolon] {./Prob_Coverage_RealTime_60m_65dB_1001.txt};
\addlegendentry{-65dBm / 1 / 60m};

\addplot[only marks,color=black, mark size =2, mark=o,mark options={solid}, mark repeat=3,mark phase=1,line width=0.5] table[x=Threshold ,y=NEH_1 ,col sep=semicolon] {./Prob_Coverage_RealTime_60m_60dB_1001.txt};
\addlegendentry{-60dBm / 1 / 60m};

\addplot[dashed,color=black, mark size =2.5, mark=triangle*,mark options={solid}, mark repeat=3,mark phase=1,line width=0.5] table[x=Threshold ,y=NEH_80 ,col sep=semicolon] {./Prob_Coverage_RealTime_60m_65dB_1080.txt};

\addplot[solid,color=black, mark size =2.5, mark=triangle*,mark options={solid}, mark repeat=3,mark phase=1,line width=0.5] table[x=Threshold ,y=NEH_80 ,col sep=semicolon] {./Prob_Coverage_Without_60m_65dB_1080.txt};

\addplot[solid,color=blue, mark size =2.5, mark=triangle*,mark options={solid}, mark repeat=3,mark phase=1,line width=0.5] table[x=Threshold ,y=NEH_80 ,col sep=semicolon] {./Prob_Coverage_Prop_60m_65dB_1080.txt};

\addplot[dashed,color=black, mark size =2, mark=triangle,mark options={solid}, mark repeat=3,mark phase=2,line width=0.5] table[x=Threshold ,y=NEH_1 ,col sep=semicolon] {./Prob_Coverage_RealTime_60m_65dB_1001.txt};

\addplot[dashed,color=black, mark size =2, mark=o,mark options={solid}, mark repeat=3,mark phase=1,line width=0.5] table[x=Threshold ,y=NEH_1 ,col sep=semicolon] {./Prob_Coverage_RealTime_60m_60dB_1001.txt};

\addplot[solid,color=black, mark size =2, mark=o,mark options={solid}, mark repeat=3,mark phase=1,line width=0.5] table[x=Threshold ,y=NEH_1 ,col sep=semicolon] {./Prob_Coverage_Without_60m_60dB_1001.txt};

\addplot[dotted,color=black, mark size =2, mark=o,mark options={solid}, mark repeat=3,mark phase=1,line width=0.5] table[x=Threshold ,y=Set1 ,col sep=semicolon] {./Prob_Coverage_OG_60_0050.txt};

\draw[red] (axis cs: 14, 0.968) ellipse (5pt and 30pt);

    \end{axis}
\end{tikzpicture}
}

 	\caption{Coverage probability (legend: $\mathsf{P_{Rx}}$ / $N_e$ / R)}
	\label{fig:coverage_Prob}
\end{figure}
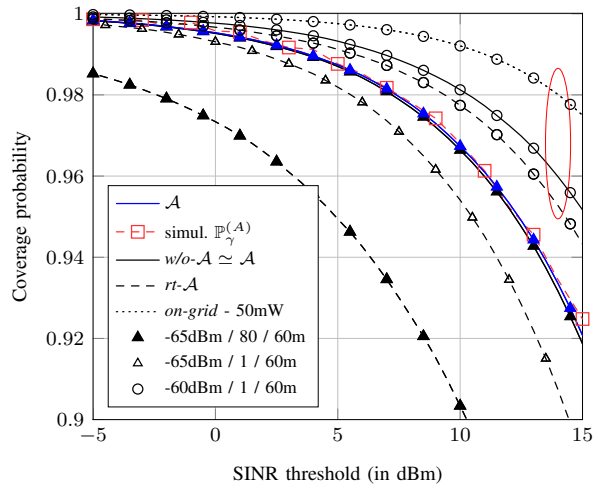 

Finally, we observed from Figure \ref{fig:coverage_Prob} that the proposed association strategy does not lead to noticeable reduction in the coverage probability compared to \textit{w/o-}$\mathcal{A}$, although it allows to serve more users (This is plotted for the case $\mathsf{P_{Rx}}=$-65dBm / $N_e =80$ / R=60m only, for more readability of the Figure). Even though a MT may associate to a farther scBS by applying the power availability checking, the resulting additional transmit power at the serving BS is not high enough to affect the interference received at users of other cells.
Such remark is also valid for the \textit{on-grid} scheme, which mostly improves the outage probability at cell-edge.
Another important outcome that can be derived from this figure regards the user selection strategy. While for both \textit{w/o-}$\mathcal{A}$ and $\mathcal{A}$, users are selecting in ascending order of power requirement, the \textit{rt-}$\mathcal{A}$ is based on a first-come first served approach. It follows that a user with very high power requirement, for example due to weak channel state, will be served as soon as sufficient energy is available, hence generating high interference for other users. This is particularly detrimental for bursty energy arrivals ($N_e = 80$).

\section{Conclusion} 
\label{sec:conclusion}

We have proposed in this paper a novel user association strategy for small-cell networks constrained by energy harvesting. It takes advantages from (i) a BS power availability checking prior to effective association, (ii) transmit power minimization, to satisfy a received power constraint and (iii) user selection at BSs  in ascending order of the  power requirement. The proposed strategy only necessitates periodical broadcast of the BS battery level and accommodates the power consumed to transmit data to all users potentially associated to the same BS.
Then, we have proposed for it a comprehensive analysis which jointly captures the effect of the battery fluctuations and the users power requirement. Simulations show that the proposed scheme significantly outperforms conventional schemes. By efficiently offloading traffic at BSs with low battery, the proposed cell association allows a more even distribution of the available energy in the network and acts as a non-direct energy-sharing mechanism. In addition, it is more robust to energy harvesting impairment and BS density reduction.

\appendices

\section{Proof of Eq. \ref{eq:estimate_Psetminus0}}
\label{appendix:estimate_Psetminus0}

From the displacement theorem as used in \cite[Lemma 1]{Blaszczyszyn2013_Infocom},
\begin{align}
& \Lambda_{MT}(p) = \lambda_{MT} \int_{\mathbb{R}^+} \mathbb{P} \left( \mathsf{P_{Rx}} \frac{l(r)}{ \chi_{0j}} \in [0,p)\right) dr 
\nonumber \\
& = \lambda_{MT} \mathbb{E}_{\chi} \left[  \pi \left(\frac{p \chi_{0j}}{\kappa \mathsf{P_{Rx}}}\right)^{\frac{2}{\alpha}} \right]
 = \lambda_{MT} \pi \left(\frac{p}{\kappa \mathsf{P_{Rx}}}\right)^{\frac{2}{\alpha}} \mathbb{E} \left[\chi_{0j}^{\frac{2}{\alpha}}\right]
\label{eq:Lambda_MT}
\end{align}
The last line comes from the independence of the shadowing attenuation on each scBS-MT link and
$\mathbb{E} \left[\chi_{0j} ^{2/\alpha}\right]$ is the fractional moment of a log-normal RV. With $\zeta = 10/ \ln(10)$,
\begin{align}
\mathbb{E} \left[\chi_{0j}^{2/\alpha}\right] = \exp \left(  \frac{2/\alpha}{\zeta} \mu + \frac{1}{2} \left(\frac{2/\alpha}{\zeta}\right)^2 \sigma^2 \right) = \Upsilon_{\chi}
\end{align}

Then, the estimate total power consumption $\widetilde{\mathsf{P}}_{k\setminus 0}^{(T)}$ is computed using the Campbell's theorem as follows:
\begin{align*}
\widetilde{\mathsf{P}}_{k\setminus 0}^{(T)}  &
= \underset{0}{\overset{p_{k0}}{\int}} p d\Lambda_{MT}(p)dp
 = \lambda_{MT} \Upsilon \frac{2 / \alpha} {2 / \alpha +1}  \left( p_{k0} \right)^{\frac{2}{\alpha}+1}.
\end{align*}

\section{Proof of Theorem \ref{prop:useful_lemma}}
\label{appendix:useful_lemma}

Let's consider the probability generating function of total power consumption $\Sigma = \underset{j \in \Phi}{\sum} \lceil p_{kj} \rceil$:

\begin{align}
& G_{\Sigma} \left(t\right) = \mathbb{E} \left( t^{\Sigma} \right)
= \underset{m=0}{\overset{\infty}{\sum}} t^m \mathbb{P} \left(\Sigma = m\right)
\\
& \text{such that} \quad  
\forall m, \quad \mathbb{P} \left(\Sigma = m\right) = \frac{1}{m!}\frac{d^m G_{\Sigma} }{dt^m} (0)
\end{align}

Since the MTs served by the typical base station scBS$^{(0)}$ form a PPP of density $\Lambda(p)$, we have:
\begin{align*}
G_{\Sigma} (t)& =  \mathbb{E} \left( t^{\left(  \sum_j \lceil p_{0j} \rceil \right)} \right)
=\mathbb{E} \left( \prod_j  t^{\lceil p_{0j} \rceil } \right)
\quad \text{with} \quad j\in \Phi \cap \left[0,P\right[
\\
& = \exp \left(- \int_{0}^{P} (1- t^{\lceil p \rceil}) d\Lambda (p)dp \right) 
\\
& = \exp \left(- \sum_{q=1}^{P} (1- t^q) \int_{q-1}^{q} d\Lambda (p)dp \right) 
\end{align*}
Let's denote $C_q = \int_{q-1}^{q} d\Lambda (p)dp$. 
Then, proceeding by induction, we get:
\begin{align*}
(0)&\left \lbrace \begin{array}{ll}
G_{\Sigma} \left(t\right)
& = \exp \left(\sum_{q=1}^{P} C_q (t^q -1) \right) 
\\
G_{\Sigma} \left(0\right)
& = \exp \left(-\sum_{q=1}^{P} C_q \right) 
\end{array}
\right.
\\
(1)&\left \lbrace \begin{array}{ll}
\frac{d G_{\Sigma}}{ dt } (t) 
&= \left(\sum_{q=1}^{P} C_q q t^{q -1}\right) G_{\Sigma} (t)
\\
\frac{d G_{\Sigma}}{ dt } (0) 
&= C_1 G_{\Sigma} (0)
\end{array}
\right.
\\
(2)&\left \lbrace \begin{array}{ll}
\frac{d^2 G_{\Sigma}}{ dt^2 } (t) 
&= \left(\sum_{q=1}^{P} C_q q t^{q -1}\right) \frac{d G_{\Sigma}}{ dt } (t) 
\\
& \quad +
\left(\sum_{q=2}^{P} C_q q (q-1) t^{q -2}\right) G_{\Sigma} (t)
\\
\frac{d^2 G_{\Sigma}}{ dt^2 } (0) 
&= C_2 G_{\Sigma} (0) + C_1 \frac{d G_{\Sigma}}{ dt } (0)
\end{array}
\right.
\end{align*}
which leads to 
\begin{align*}
\frac{d^m G_{\Sigma}}{ dt^m } (0) = \underset{q=1}{\overset{m^{\star}}{\sum}} \binom{m-1}{q-1} C_q \; q! \; \frac{d^{m-q} G_{\Sigma}}{ dt^{m-q} } (0) 
\end{align*}
and completes the proof.

\section{Proof of Proposition \ref{prop:coverage_probability}}
\label{appendix:coverage_prob}

We first focus on the coverage probability. We have:
\begin{align}
\mathbb{P}_\gamma & = \mathbb{P} \left(  \frac{\mathsf{P_{Rx}} \vert h_{1,0} \vert^2}{I} \geq T \right) \nonumber \\
&=
\mathbb{E} \left[ \mathbb{P}_h \left( \vert h_{1,0} \vert^2 \geq \frac{T}{\mathsf{P_{Rx}} } \sum_{m=1}^{\mathsf{L}} \sum_{l=m}^{\mathsf{L}} {\sum}_{\Phi_{m,l}^{(I)}}
 \frac{m}{ N_{RB}} \frac{\chi_{k,0}}{l(r_{k,0})}\vert h_{k,0} \vert^2 \right) \right]
\nonumber \\ 
&= \mathbb{E} \left[ \exp \left( - \sum_{m=1}^{\mathsf{L}} \sum_{l=m}^{\mathsf{L}} {\sum}_{\Phi_{m,l}^{(I)}} \nu \frac{T}{\mathsf{P_{Rx}} } \frac{m}{ N_{RB}} \frac{\chi_{k,0}}{l(r_{k,0})}\vert h_{k,0} \vert^2 \right) \right]
\nonumber \\ 
&= \prod_{m=1}^{\mathsf{L}} \prod_{l=m}^{\mathsf{L}} \mathbb{E} \left[ \exp \left( -  {\sum}_{\Phi_{m,l}^{(I)}} \nu \frac{T}{\mathsf{P_{Rx}} } \frac{m}{ N_{RB}} \frac{\chi_{k,0}}{l(r_{k,0})}\vert h_{k,0} \vert^2 \right) \right]
\nonumber \\
&= \prod_{m=1}^{\mathsf{L}} \prod_{l=m}^{\mathsf{L}}  \mathcal{L}_{m,l} \left(\nu \frac{T}{\mathsf{P_{Rx}} }\right)
\end{align}
The Laplace transform of the total interference received from scBSs in $\Phi_{m,l}^{(I)}$ is given by
\begin{align}
\hspace*{-5pt}\mathcal{L}_{m,l} \left(s\right) 
&= \mathbb{E} \left[ {\prod}_{\Phi_{m,l}^{(I)}}  \exp \left( - s \frac{m}{ N_{RB}} \frac{\chi_{k,0}}{l(r_{k,0})}\vert h_{k,0} \vert^2 \right) \right]
\nonumber \\
&= \mathbb{E}_{\Phi_{m,l}^{(I)}} \left[ {\prod}_{\Phi_{m,l}^{(I)}}  \mathbb{E}_h \left[\exp \left( - s \frac{m}{ N_{RB}} \frac{\chi_{k,0}}{l(r_{k,0})}\vert h_{k,0} \vert^2 \right) \right]\right]
\nonumber \\
&= \mathbb{E}_{\Phi_{m,l}^{(I)}} \left[ {\prod}_{\Phi_{m,l}^{(I)}} \frac{\nu}{\nu + s \frac{m}{ N_{RB}} \frac{\chi_{k,0}}{l(r_{k,0})}} \right]
\end{align}
Let's denote $z_k =\frac{ N_{RB}}{m}  \frac{l(r_{k,0})}{\chi_{k,0}} = \frac{ N_{RB}}{m\mathsf{P_{Rx}} } p_{k0}, \forall k$, which refers to the inverse of the power received from the interfering scBS$_k$. The set of $z_k$ thus forms a displaced PPP of $\Phi_{m,l}^{(I)}$. Its intensity is computed using a proof similar to Appendix \ref{appendix:estimate_Psetminus0} as follows:
\begin{align}
& \Lambda_{m,l}^{(I)}(z) = \rho_{m,l} \lambda_{B} \Upsilon_m z^{\frac{2}{\alpha}}
\end{align}
where $\rho_{m,l}$ and $\Upsilon_m$ are as in Proposition \ref{prop:coverage_probability}. In addition, let's consider scBS$_j$ the closest interfering base station of $\Phi_{m,l}^{(I)}$ (closest in terms of both path-loss and shadowing).
If scBS$_j$ is available for the typical user MT$_0$, we know for sure that $z_{j} > \frac{1}{\mathsf{P_{Rx}}}$. The availability criteria also implies that $z_{j} < \frac{ N_{RB}}{m\mathsf{P_{Rx}} } p_l^{(\text{cov})}$. 
If scBS$_j$ is not available for MT$_0$, then $z_{j} > \frac{ N_{RB}}{m\mathsf{P_{Rx}} } p_l^{(\text{cov})}$. 
Defining $z_{m,l} = \min \left(\frac{1}{\mathsf{P_{Rx}}}, \frac{ N_{RB}}{m\mathsf{P_{Rx}} } p_l^{(\text{cov})}\right)$ and $ u = \frac{\nu}{s} z$, we get:
\begin{align}
\mathcal{L}_{m,l} \left(s\right) 
&= \mathbb{E}_{\Phi_{m,l}^{(I)}} \left[ {\prod}_{\Phi_{m,l}^{(I)}} \frac{\nu}{\nu + s / z_k} \right]
\nonumber \\
&= \exp \left( - \int_{z_{m,l}}^{\infty} \left(1- \frac{\nu}{\nu + s / z}\right) d \Lambda_{m,l}^{(I)}(z) dz
\right)
\nonumber \\
&= \exp \left( -  \rho_{m,l} \lambda_{B} \Upsilon_m \frac{2}{\alpha} \left(\frac{s}{\nu}\right)^{\frac{2}{\alpha}} \int_{u_{m,l}}^{\infty} \frac{u^{\frac{2}{\alpha}-1}}{1+u}  du
\right)
\end{align}
which concludes the proof.

\fontsize{10}{10}
\selectfont

\bibliographystyle{IEEEtranN}
{\footnotesize 
\bibliography{biblio_EHBS_journal1}
}

\end{document}